\documentclass[11pt,a4paper]{article}

\usepackage{amsthm, amsfonts,amssymb,euscript}
\usepackage{amsmath}
\usepackage{bm}
\usepackage{amssymb}
\usepackage{amsthm}
\usepackage{graphicx}
\usepackage{caption}
\usepackage{subcaption}
\usepackage{amsfonts}
\usepackage{float}
\usepackage{color}
\usepackage{slashed}

\usepackage[usenames,dvipsnames,svgnames,table]{xcolor}
\usepackage[colorlinks=true]{hyperref}
\hypersetup{linkcolor=BrickRed, urlcolor=green, citecolor=blue, linktoc=page}

\numberwithin{equation}{section}

\newtheorem*{proposition*}{Proposition}
\newtheorem*{theorem*}{Theorem}
\newtheorem*{conjecture*}{Conjecture}
\newtheorem*{claim*}{Claim}
\newtheorem*{lemma*}{Lemma}
\newtheorem*{corollary*}{Corollary}
\newtheorem{theorem}{Theorem}[section]
\newtheorem{proposition}[theorem]{Proposition}

\newtheorem{lemma}[theorem]{Lemma}

\newtheorem*{definition*}{Definition}

\newtheorem*{assumption*}{\mathcal{A}ssumption}

\newtheorem*{remark*}{Remark}

\newcommand{\s}{\mathbb{S}}

\allowdisplaybreaks

\begin{document}
\title{Asymptotics for scalar perturbations from a neighborhood of the bifurcation sphere}
\author{Y. Angelopoulos, S. Aretakis, and D. Gajic}
\date{June 8, 2018}

\maketitle
\begin{abstract}

In our previous work {\small [{\sc Angelopoulos, Y., Aretakis, S. and Gajic, D.} Late-time asymptotics for the wave equation on spherically symmetric stationary backgrounds. {\em Advances in Mathematics 323\/} (2018), 529--621.]} we showed that the coefficient in the precise leading-order late-time asymptotics for solutions to the wave equation with smooth, compactly supported initial data on Schwarzschild backgrounds is proportional to the time-inverted Newman--Penrose constant (TINP), that is the Newman--Penrose constant of the associated time integral. The time integral (and hence the TINP constant) is canonically defined in the domain of dependence of any Cauchy hypersurface along which the stationary Killing field is non-vanishing.  As a result, an explicit expression of the late-time polynomial tails was obtained in terms of initial data on Cauchy hypersurfaces intersecting the future event horizon to the future of the bifurcation sphere.

In this paper, we extend the above result to Cauchy hypersurfaces intersecting the bifurcation sphere via a novel geometric interpretation of the TINP constant in terms of a modified gradient flux on Cauchy hypersurfaces. We show, without appealing to the time integral construction, that a general conservation law holds for these gradient fluxes. This allows us to express the TINP constant in terms of initial data on Cauchy hypersurfaces for which the time integral construction breaks down.

\end{abstract}

\tableofcontents

\section {Introduction}
\label{introduction}

\subsection{Introduction and background}
\label{intro}

Late-time asymptotics for solutions to the wave equation
\begin{equation}
\Box_{g}\psi=0
\label{we}
\end{equation}
on globally hyperbolic Lorentzian manifolds $(\mathcal{M},g)$ have important applications in the study of problems that arise in general relativity such as 1) the black hole stability problem, 2) the dynamics of black hole interiors and strong cosmic censorship and 3) the propagation of gravitational waves. 

The existence of late-time \textit{polynomial} tails for solutions to the wave equation with smooth, \textit{compactly supported} initial data on \textit{curved} spacetimes\footnote{Recall that, in view of Huygens' principle, solutions with compactly supported initial data on the flat Minkowski spacetime $\mathbb{R}^{3+1}$ identically vanish inside a ball of arbitrary large radius after a sufficiently large time. } was first heuristically obtained by Price \cite{Price1972} in 1972. Specifically, the work in \cite{Price1972} suggests that on Schwarzschild spacetimes $(\mathcal{M}_M,g_M)$, with $M>0$, such solutions $\psi$ have the following asymptotic behavior in time as $\tau\rightarrow \infty$:
\begin{equation}
\psi|_{r=r_0}(\tau,r=r_0,\theta,\varphi)\sim \frac{1}{\tau^3} \text{ along the } r=r_0>2M \text{  hypersurfaces} 
\label{eq:1}
\end{equation}
away from the event horizon $\mathcal{H}^{+}=\{r=2M\}$. Subsequent work by Leaver\footnote{The work of Leaver moreover relates the existence of power law tails in the late-time asymptotics to the presence of a branch cut in the Laplace transformed Green's function corresponding to the equations satisfied by fixed spherical harmonic modes.} \cite{leaver} and Gundlach, Price and Pullin \cite{CGRPJP94b} suggested the following additional asymptotics:
\begin{equation}
\psi|_{\mathcal{H}^{+}}(\tau,r=2M,\theta,\varphi) \sim \frac{1}{\tau^3} \text{ along the event horizon }\mathcal{H}^{+},
\label{eq:2}
\end{equation}
and
\begin{equation}
r\psi|_{\mathcal{I}^{+}}(\tau,r=\infty,\theta,\varphi) \sim \frac{1}{\tau^2} \text{ along the null infinity }\mathcal{I}^{+},
\label{eq:3}
\end{equation}
which would imply that the late-time tails are ``radiative''.
\begin{figure}[H]
\begin{center}
\includegraphics[width=6cm]{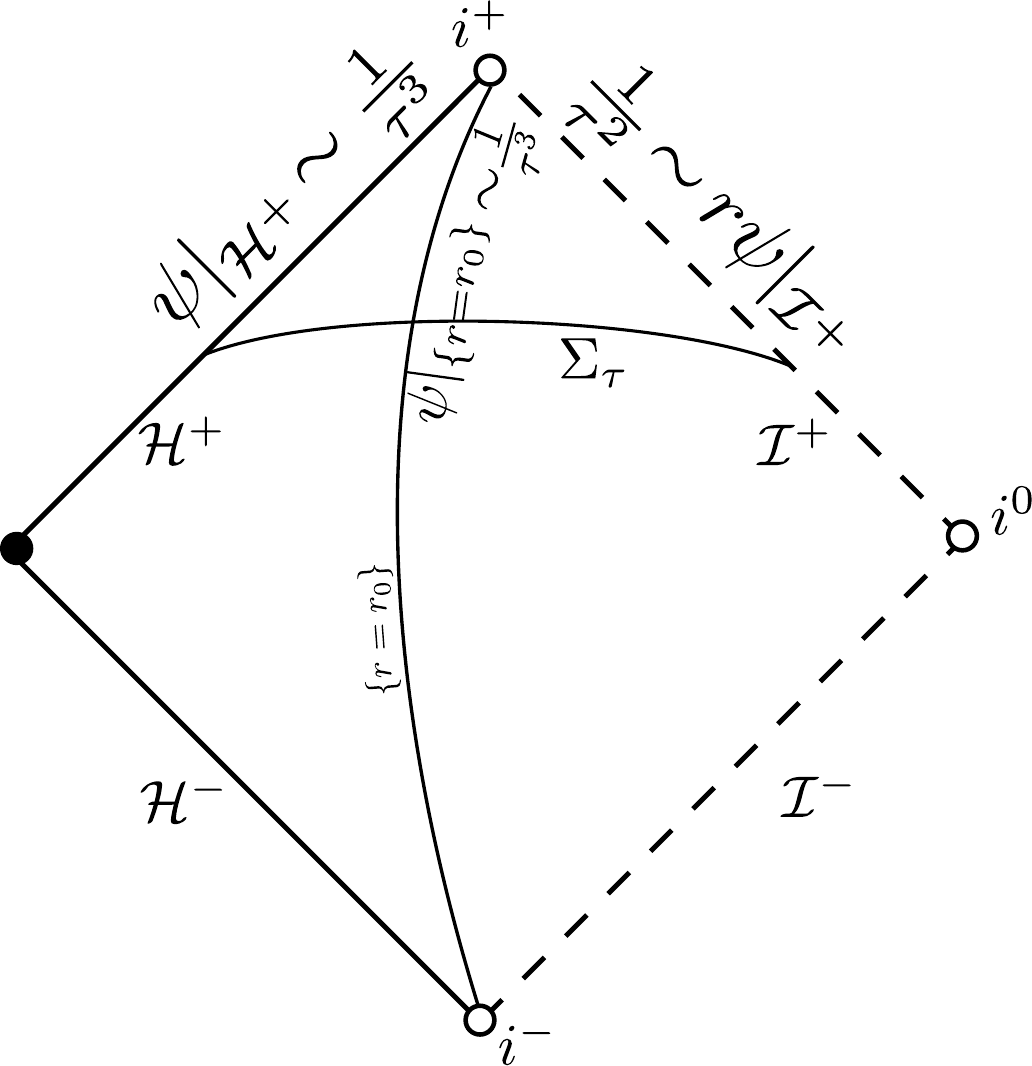}
\caption{\label{fig:1}Price's asymptotics.}
\end{center}
\end{figure}
\vspace{-0.9cm}
Here $\tau$ is an appropriate ``time'' parameter that is comparable to the $t$ coordinate away from the event horizon and null infinity, the level sets $\Sigma_\tau$ of which are spacelike hypersurfaces that cross the event horizon and terminate at future null infinity. There has been a large number of other works in the physics literature elucidating novel aspects of the behavior of tails of scalar (and electromagnetic and gravitational) fields from a heuristic or numerical point of view; see for example \cite{cunmonpr, Gomez1994, clsy95, karzdma, Bicak1972, hm2012, ori2013, sela} for tails on spherically symmetric spacetime backgrounds. For works on tails in Kerr spacetimes, see \cite{LBAO99, pricekerr, burkokerr, zimmerman1} and the references therein.

Note that in view of their asymptotic character, \eqref{eq:1}, \eqref{eq:2} and \eqref{eq:3} do \underline{not} provide estimates for the size of solutions for \textbf{all} times. Furthermore, they are restricted to fixed spherical harmonic modes and do not address the decay behavior when summing over all the modes. These issues have been extensively addressed by rigorous mathematical works in the past, which derived sharp global quantitative bounds for solutions to the wave equation on general black hole spacetimes. See, for instance, \cite{ part3,dhr-teukolsky-kerr, lecturesMD, MDIR05, Dafermos2016,    moschidis1,  volker1, tataru3, metal, blukerr, dssprice, kro} and references therein in the asymptotically flat setting and see \cite{peter2, semyon1, other1, gusmu1, gusmu2} in the asymptotically de Sitter and anti de Sitter setting. These works have introduced numerous new insights and techniques in the study of the long time behavior of solutions to the wave equation on curved backgrounds. For example, the role of the \textit{redshift effect} as a stabilizing mechanism was first understood in \cite{redshift} and difficulties pertaining to null infinity were addressed in \cite{newmethod,moschidis1}. We also refer to \cite{baskinw, baskinwang, baskin16} for results regarding the existence of an asymptotic expansion in time for solutions to the wave equation on certain asymptotically flat spacetimes. Recent very important developments in the context of stability problems include \cite{HV2016,Dafermos2016,klainerman17,moschidisads}. 

 A special case that has recently attracted strong interest is that of extremal black holes for which the asymptotic terms are very different in view of instabilities of derivatives of the scalar fields (see for instance \cite{aretakis1,aretakis2,aretakis3,aretakis4,aag1,aretakis2012,aretakis2013,harvey2013,harveyeffective,zimmerman1,zimmerman2}). 

On other hand, the above works do not yield global \textit{pointwise lower bounds}\footnote{Note that sharpness of the decay rates of \emph{energy fluxes} along the event horizon was first established by Luk and Oh \cite{luk2015}.} on the size of solutions to \eqref{we}. Furthermore, global quantitative upper bounds or asymptotic expressions of the type \eqref{eq:1}, \eqref{eq:2} and \eqref{eq:3} do not provide the exact leading-order late-time asymptotic terms in terms of explicit expressions of the initial data. In other words, these works do not recover how the coefficients appearing in front of the dominant terms in the asymptotic expansion in time depend precisely on initial data. The derivation of the \emph{exact} late-time asymptotic terms is very important in studying the behavior of scalar fields in the interior of extremal black hole spacetimes \cite{gajic, gajic2, dejanjon1} and also in the Schwarzschild black hole interior \cite{gregjan}. Note that lower bounds for the decay rate of scalar fields moreover play an important role when studying the properties of the interiors of dynamical sub-extremal black holes. This was first shown rigorously in \cite{MD03,MD05c};  see also \cite{MD12, luk2015, LukSbierski2016, DafShl2016, Hintz2015, Franzen2014, Luk2016a,Luk2016b} for related results in the interior of sub-extremal black holes.

The above issue was resolved in our recent work \cite{paper2} where global quantitative estimates were obtained of the form
\begin{equation}
\left|\psi(\tau, r_0,\theta,\varphi)- \boldsymbol{Q}_{\Sigma_0,r_0}[\psi]\cdot\frac{1}{\tau^3}\right| \leq C_{r_0} \cdot \sqrt{E_{\Sigma_0}[\psi]}\cdot \frac{1}{\tau^{3+\epsilon}} 
\label{eq:ours}
\end{equation}
along the hypersurface $\{r = r_0\}\cap \mathcal{J}^{+}(\Sigma_0)$ in Schwarzschild (and a more general class of spherically symmetric, asymptotically flat spacetimes), with $\sqrt{E_{\Sigma_0}[\psi]}$ an initial data norm and $\epsilon>0$ and $C_{r_0}$ positive constants (independent of the initial data for $\psi$). We denote by $\mathcal{J}^{+}(\Sigma_0)$ the future of the Cauchy hypersurface $\Sigma_0$. In particular, the estimate \eqref{eq:ours} holds along the future event horizon $\mathcal{H}^{+}$ where $r_0=2M$.  Here $\tau$ is an appropriate time parameter in $\mathcal{J}^{+}(\Sigma_0)$. In \cite{paper2} we only considered initial hypersurfaces which cross the event horizon \textit{to the future of the bifurcation sphere} and terminate at null infinity. Estimate \eqref{eq:ours} holds for all solutions $\psi$ to \eqref{we} which arise from smooth compactly supported initial data\footnote{Moreover, similar estimates were shown for more general initial data decaying in $r$.} on  $\Sigma_0$. 
\begin{figure}[H]
\begin{center}
\includegraphics[width=6cm]{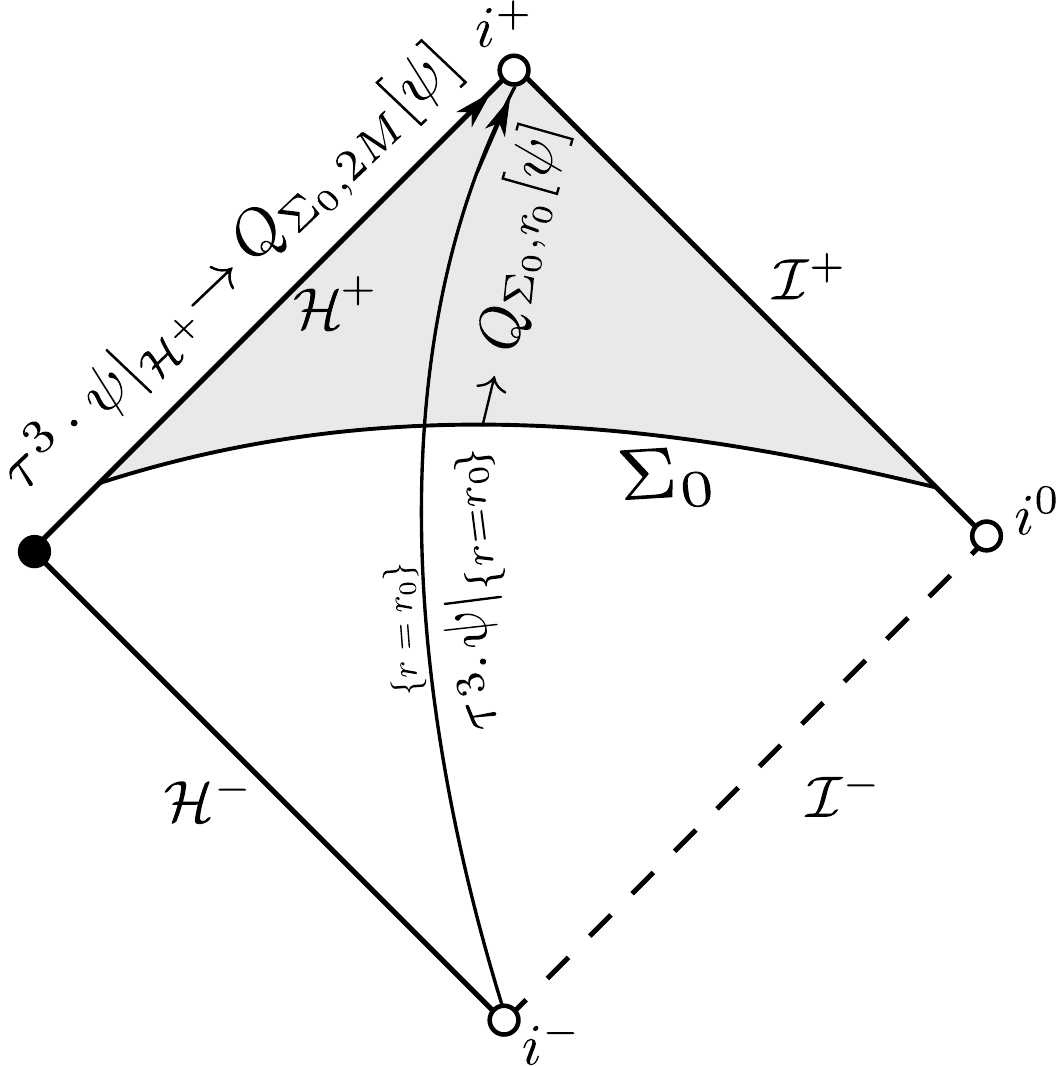}
\caption{\label{fig:1a}The estimate \eqref{eq:ours} provides the asymptotics limits along the hypersurfaces $\{r=r_0\}$.}
\end{center}
\end{figure}
\vspace{-.9cm}

As a result of the method of \cite{paper2}, \textit{an explicit expression of the coefficient of the leading order term $\boldsymbol{Q}_{\Sigma_0,r_0}[\psi]$ in \eqref{eq:ours} was derived in terms of the initial data of $\psi$ on $\Sigma_0$.} In fact, it was shown that $\boldsymbol{Q}_{\Sigma_0,r_0}[\psi]$\textbf{ is independent of} $r_0$:
\[\boldsymbol{Q}_{\Sigma_0,r_0}[\psi]=\boldsymbol{Q}_{\Sigma_0}[\psi]. \]
Furthermore, it was shown that the radiation field satisfies
\begin{equation}
\left|r\psi|_{\mathcal{I}^{+}}(\tau,\theta,\varphi)- \frac{1}{4}\boldsymbol{Q}_{\Sigma_0}[\psi]\cdot\frac{1}{\tau^2}\right| \leq C \cdot \sqrt{E_{\Sigma_0}[\psi]}\cdot \frac{1}{\tau^{2+\epsilon}} 
\label{eq:oursradiation}
\end{equation}
along future null infinity $\mathcal{I}^{+}\cap \mathcal{J}^{+}(\Sigma_0)$.
\begin{figure}[H]
\begin{center}
\includegraphics[width=12cm]{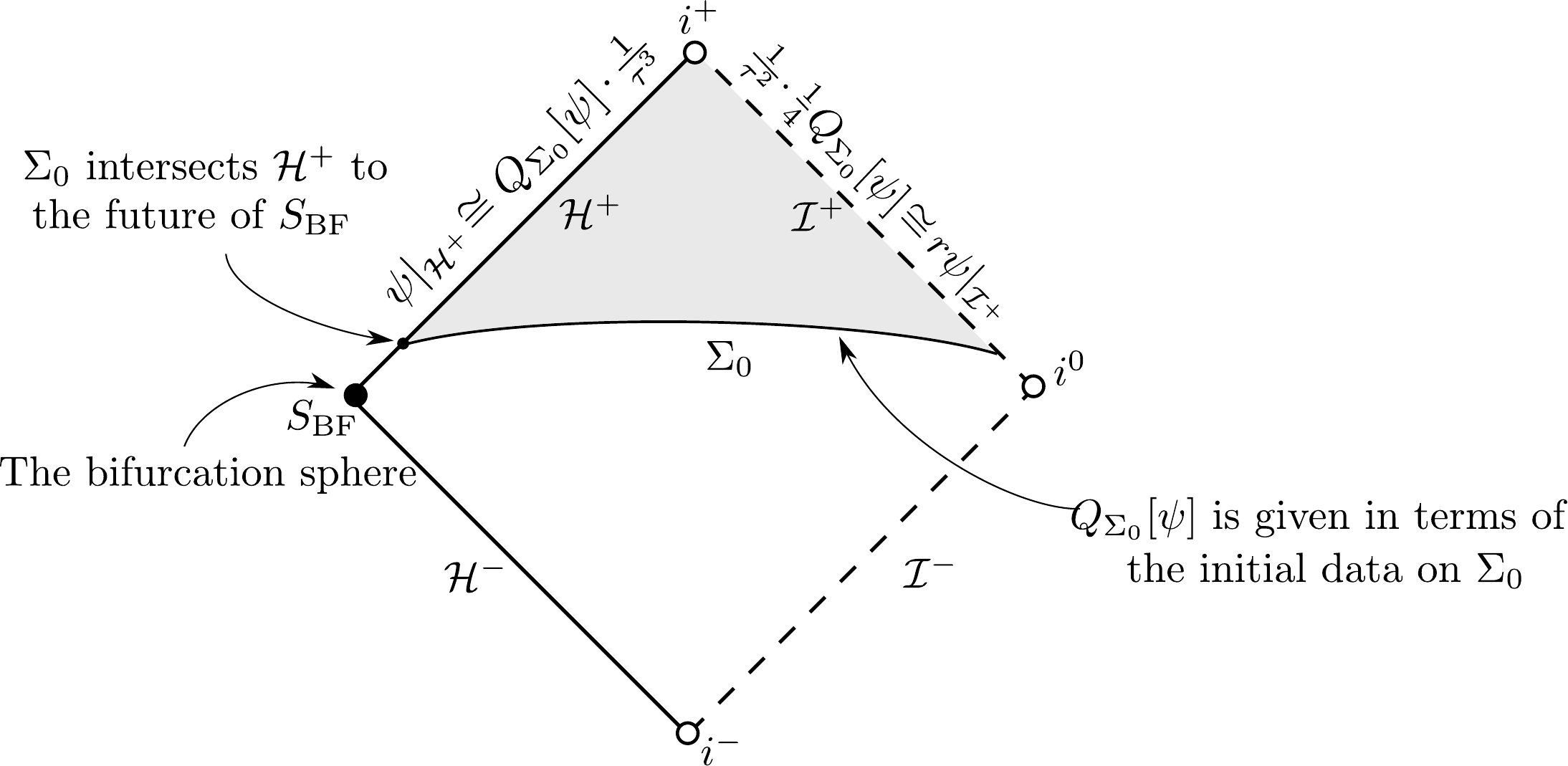}
\caption{\label{fig:2}Precise late-time asymptotics with $Q_{\Sigma_0}[\psi]$ given by the initial data of $\psi$ on a Cauchy hypersurface $\Sigma_0$ to the future of the bifurcation sphere $S_{\text{BF}}$.}
\end{center}
\end{figure}
\vspace{-.5cm}
 This work rigorously showed that \textbf{the precise fall-off the scalar fields depends on the profile of the initial data} and hence is not a universal property of scalar fields due to the background backscattering. 

In fact, in a recent paper \cite{logasymptotics}, we were able to obtain the second-order term in the asymptotic expansion of the radiation field along future null infinity which arises as a logarithmic correction to \eqref{eq:oursradiation} and show that the corresponding coefficient is proportional to the ADM mass $M$ and again to $\boldsymbol{Q}_{\Sigma_0}[\psi]$. For spherically symmetric initial data, we moreover provided in \cite{logasymptotics} the precise dependence on initial data of the \textbf{full} asymptotic expansion of $\psi$ and its radiation field $r\psi|_{\mathcal{I}^{+}}$. Both \cite{paper2,logasymptotics} used purely \textit{physical space techniques}, instead of Fourier analytic methods, and described the origin of the polynomial tails on black hole backgrounds in terms of physical space quantities.

\textit{The estimates \eqref{eq:ours} and \eqref{eq:oursradiation} provided the first rigorous confirmation of the asymptotic statements \eqref{eq:1}, \eqref{eq:2} and \eqref{eq:3}. }In particular, they provided the first global pointwise lower bounds on the scalar fields and their radiation fields. Since those bounds are determined in terms of the quantity $\boldsymbol{Q}_{\Sigma_0}[\psi]$ of the initial data, we obtained as an immediate application a characterization of all smooth, compactly supported initial data which produce solutions to \eqref{we} which decay in time \textbf{exactly} like $\frac{1}{\tau^3}$ to leading order.
It is clear, therefore, that is of great importance, to single out the exact expressions of the initial data which provide the dominant terms in the evolution of the scalar fields.

As was mentioned above, the results in \cite{paper2} hold for initial Cauchy hypersurfaces which emanate from a section of the future event horizon which lies strictly in the future of the bifurcation sphere (see Figure \ref{fig:2}). This restriction was necessary so that in $\mathcal{J}^{+}(\Sigma_0)$, the region under consideration, the stationary Killing vector field $T=\partial_t$ is non-vanishing.

In this paper, we obtain the coefficient of the leading-order late-time asymptotic terms for solutions with smooth compactly supported initial data on \textit{Cauchy hypersurfaces which pass through the bifurcation sphere}. For Schwarzschild backgrounds, an example of such a hypersurface is given by $\{t=0\}$. As we shall see, there are various qualitative differences in this case compared to the case studied in \cite{paper2}. 
\begin{figure}[H]
\begin{center}
\includegraphics[width=5.4cm]{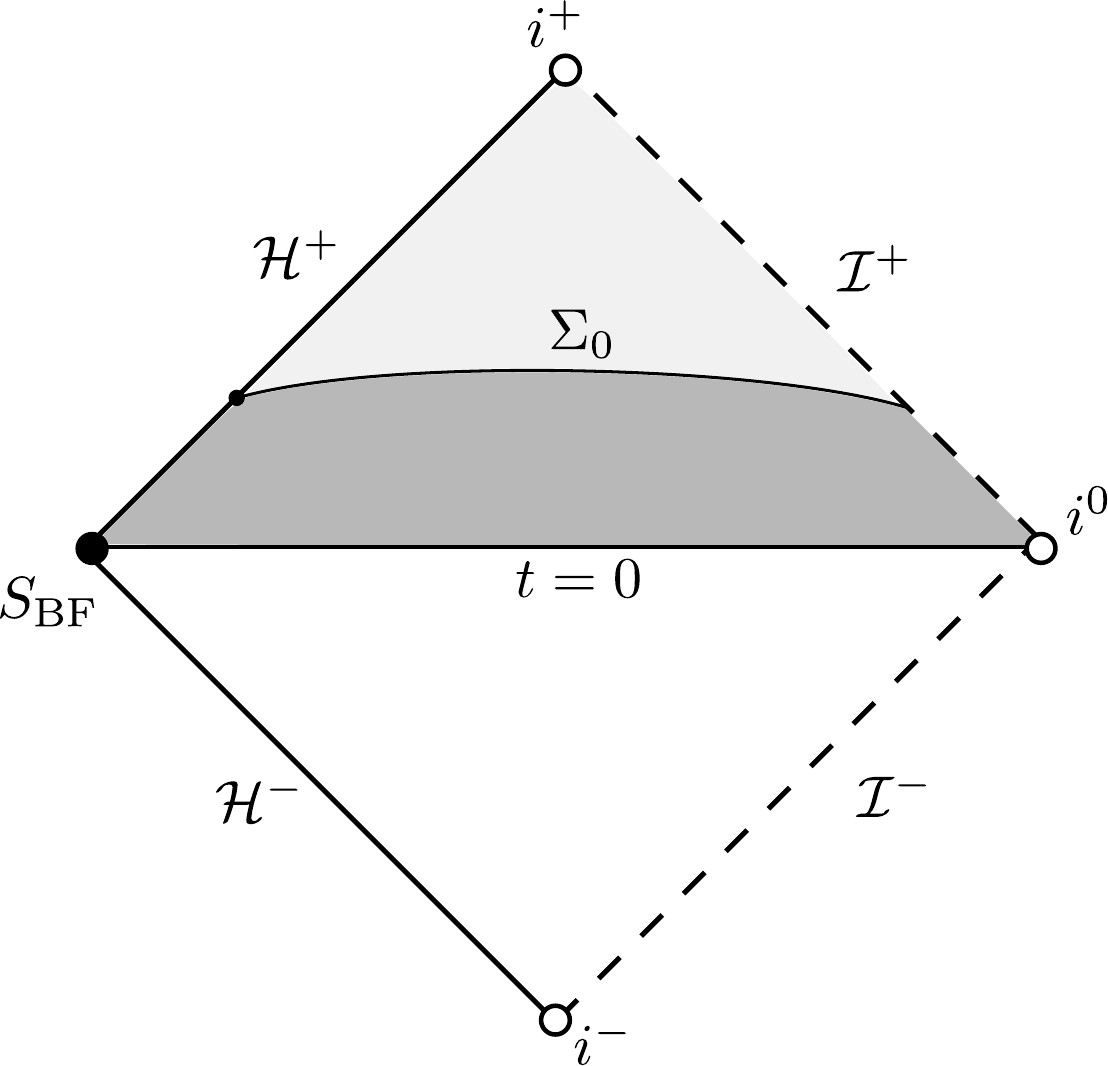}
\caption{\label{fig:3}To find the asymptotics for scalar fields with initial data on $\{t=0\}$ we need to express $\boldsymbol{Q}_{\Sigma_0}[\psi]$ in terms of the initial data on $\{t=0\}$.}
\end{center}
\end{figure}
\vspace{-0.9cm}
Before we present our method and the new results, we provide a brief review of the time integral construction which played a crucial role in \cite{paper2} and allowed us to derive the coefficient $\boldsymbol{Q}_{\Sigma_0}[\psi]$ in the asymptotic expansion in terms of the initial data on $\Sigma_0$ in the case where $\Sigma_0$ does not pass through the bifurcation sphere.

\subsection{The time integral construction and the TINP constant}
\label{sec:ReviewOfTheTimeIntegralConstruction}

The time integral construction of \cite{paper2} concerns the spherical mean
\[\psi_0=\frac{1}{4\pi}\int_{\mathbb{S}^2}\psi\,d\omega,\]
where $d\omega=\sin \theta d\theta d\varphi$.

Indeed, it was shown in \cite{paper1}  that the projection 
\[\psi_1=\psi-\frac{1}{4\pi}\int_{\mathbb{S}^2}\psi \,d\omega\]
decays at least like $\tau^{-3.5+\epsilon}$ (with arbitrarily small $\epsilon>0$) and hence does not contribute to the leading order terms in the late-time asymptotics\footnote{A similar result holds for the radiation field of the projection $\psi_1$.}. 

Given a smooth solution $\psi$ to \eqref{we}, we want to find a smooth solution $\psi^{(1)}$ to \eqref{we} such that
\begin{equation}
T\psi^{(1)}=\psi_0
\label{ti}
\end{equation}
in the future $\mathcal{J}^{+}(\Sigma_0)$ of a Cauchy hypersurface $\Sigma$ which intersects the event horizon strictly to the future of the bifurcation sphere.   It is important to restrict to such hypersurfaces since we have
\begin{equation}
T\neq 0
\label{tnonzero}
\end{equation}
in $\mathcal{J}^{+}(\Sigma_0)$. 
\begin{figure}[H]
\begin{center}
\includegraphics[width=7cm]{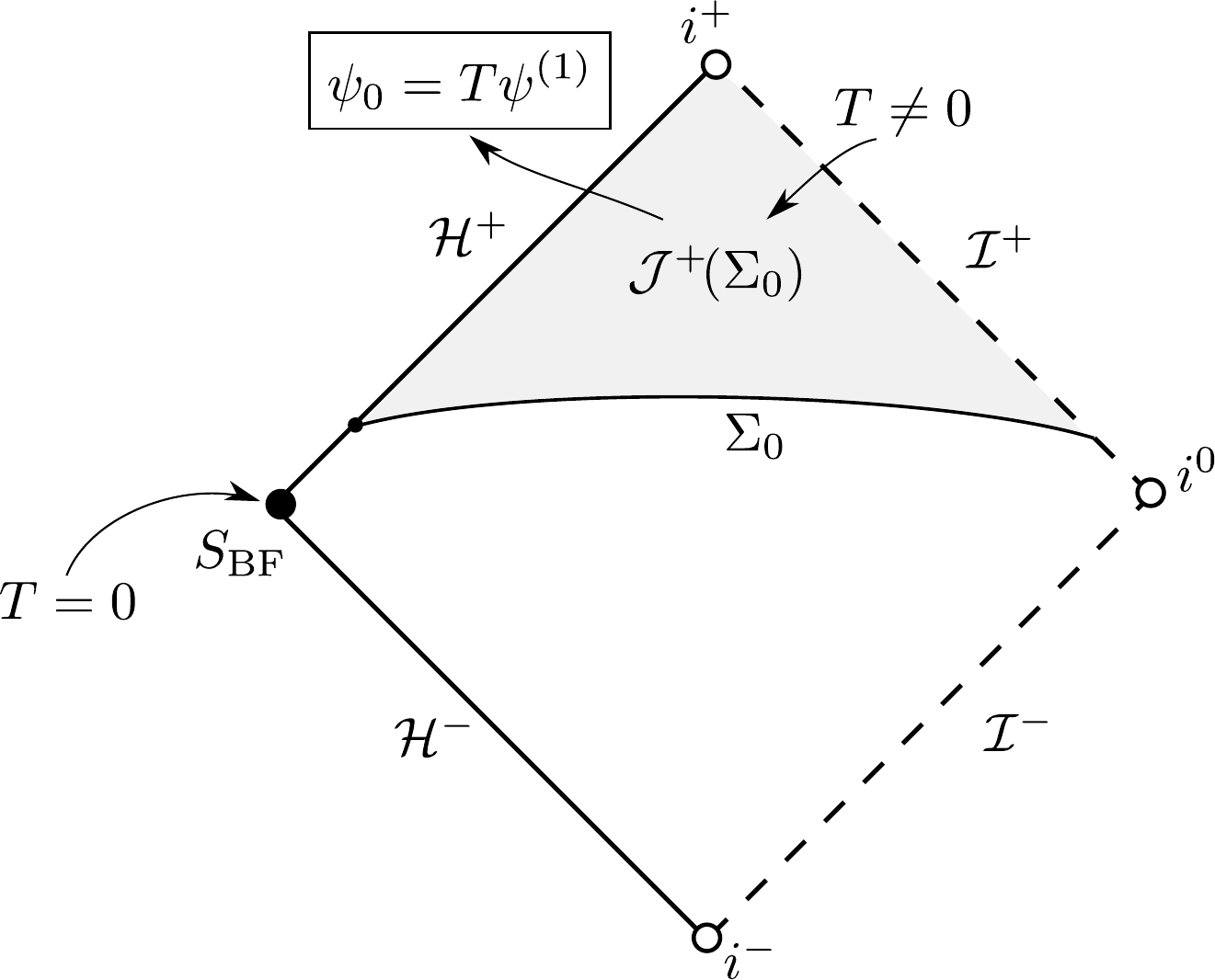}
\caption{\label{fig:4}Time inversion in the region $\mathcal{J}^{+}(\Sigma_0)$.}
\end{center}
\end{figure}
\vspace{-0.9cm}
On the other hand, the stationary Killing field $T=0$ on the bifurcation sphere on Schwarzschild spacetimes\footnote{More generally, the Killing vector field $T$ corresponding to stationarity also vanishes on the bifurcation sphere of all sub-extremal Reissner--Nordstr\"{o}m spacetimes.} and hence, if $\Sigma_0$ intersected the bifurcation sphere then we would not be able to invert the operator $T$ without imposing additional conditions on $\psi_0$. 

Nonetheless,\textit{ there is another obstruction to inverting $T$.}  This obstruction originates from the far-away region and specifically from the existence of \textit{a conservation law along null infinity}. Consider the standard outgoing Eddington--Finkelstein coordinates $(u,r,\omega)$ (with $\omega \in\mathbb{S}^2$) and the function $I_{0}[{\psi}](u)$ on the null infinity $\mathcal{I}^{+}$ given by\footnote{The derivative $\partial_r$ is taken with respect to the $(u,r,\omega)$ coordinate system.}
\[I_{0}[{\psi}](u)=\frac{1}{4\pi}\lim_{r\rightarrow\infty} \int_{\mathbb{S}^{2}}r^2 \partial_r (r{\psi}) (u,r,\omega) \, d\omega. \]
It turns out that if ${\psi}$ solves \eqref{we}, then the function $I_{0}[{\psi}](u)$ is constant, that is independent of $u$. This yields a conservation law along $\mathcal{I}^{+}$. The associated constant 
\begin{equation}
I_{0}[{\psi}]:=I_{0}[{\psi}](u)
\label{np}
\end{equation}
is called the Newman--Penrose constant of ${\psi}$.
\begin{figure}[H]
\begin{center}
\includegraphics[width=6.5cm]{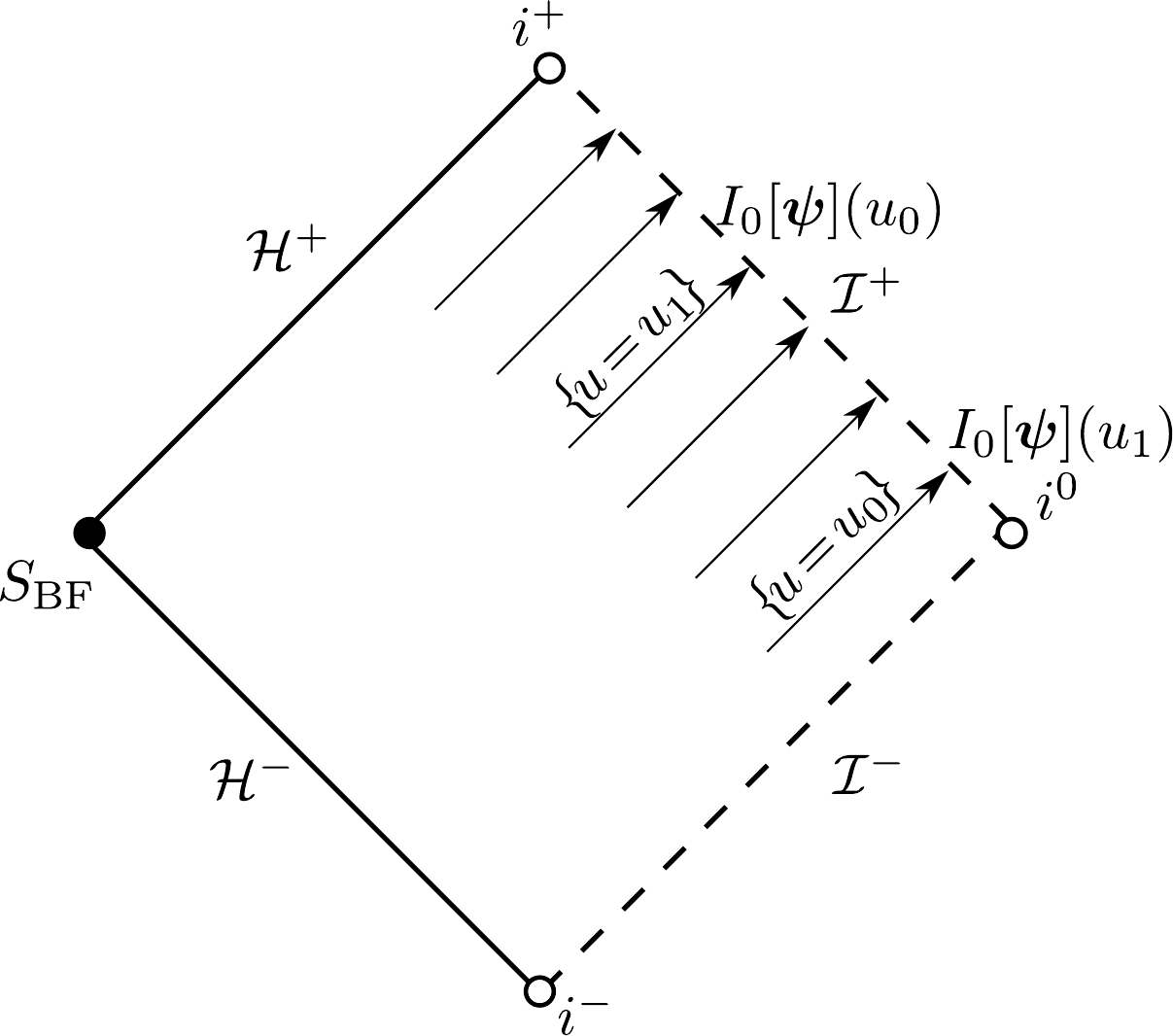}
\caption{\label{fig:5}The Newman--Penrose constant on future null infinity.}
\end{center}
\end{figure}
\vspace{-0.9cm} \textit{The Newman--Penrose constant (and its conservation law along $\mathcal{I}^{+}$) is an obstruction to the invertibility of $T$}.\footnote{We take the domain $T$ to be the space of smooth solutions to the wave equation \eqref{we} with a well-defined Newman--Penrose constant. } Indeed, if the Newman--Penrose constant of ${\psi}$ is well defined (that is the limit of the conformal derivative $r^2\partial_r(r{\psi})$ is bounded on $\Sigma_0$) then 
\[I_0[T{\psi}]=0.\]
Hence, a solution $\psi$ to the wave equation \eqref{we} is not in the range of the operator $T$ unless its Newman--Penrose constant vanishes! This obstruction is present for all asymptotically flat spacetimes.

 If we consider smooth initial data for $\psi$ on $\Sigma_0$ with vanishing Newman--Penrose constant
\[I_{0}[\psi]=0.\] 
such that in fact
\begin{equation}
\lim_{r\rightarrow \infty} \int_{\mathbb{S}^2} r^3 \partial_r(r\psi)|_{\Sigma_0}\,d\omega<\infty 
\label{frompaper2}
\end{equation}then by Proposition 9.1 of \cite{paper2} there is a \textbf{unique} smooth spherically symmetric\footnote{Since $\psi_0$ is spherically symmetric clearly it suffices to look for spherically symmetric solutions $\psi^{(1)}$ satisfying \eqref{ti}} solution
\[\psi^{(1)}:\mathcal{J}^{+}(\Sigma_0)\rightarrow \mathbb{R} \]
of the wave equation \eqref{we} that decays along the Cauchy hypersurface $\Sigma_0$:
\begin{enumerate}
	\item $\lim_{r\rightarrow \infty}\psi^{(1)}|_{\Sigma_0}=0$,
	\item $\lim_{r\rightarrow \infty}r^2\partial_r\psi^{(1)}|_{\Sigma_0}<\infty$
\end{enumerate}
satisfying
\[T\psi^{(1)}=\frac{1}{4\pi}\int_{\mathbb{S}^2} \psi\,d\omega\]
everywhere in $\mathcal{J}^{+}(\Sigma_0)$.

The solution $\psi^{(1)}$ is called the \textbf{time integral} of the sperical mean $\psi_0$ of $\psi$. The Newman--Penrose constant $I_0[\psi^{(1)}]$ of $\psi^{(1)}$ is well-defined and can be explicitly computed in terms of the initial data of $\psi$ on $\Sigma_0$.\begin{figure}[H]
\begin{center}
\includegraphics[width=9.8cm]{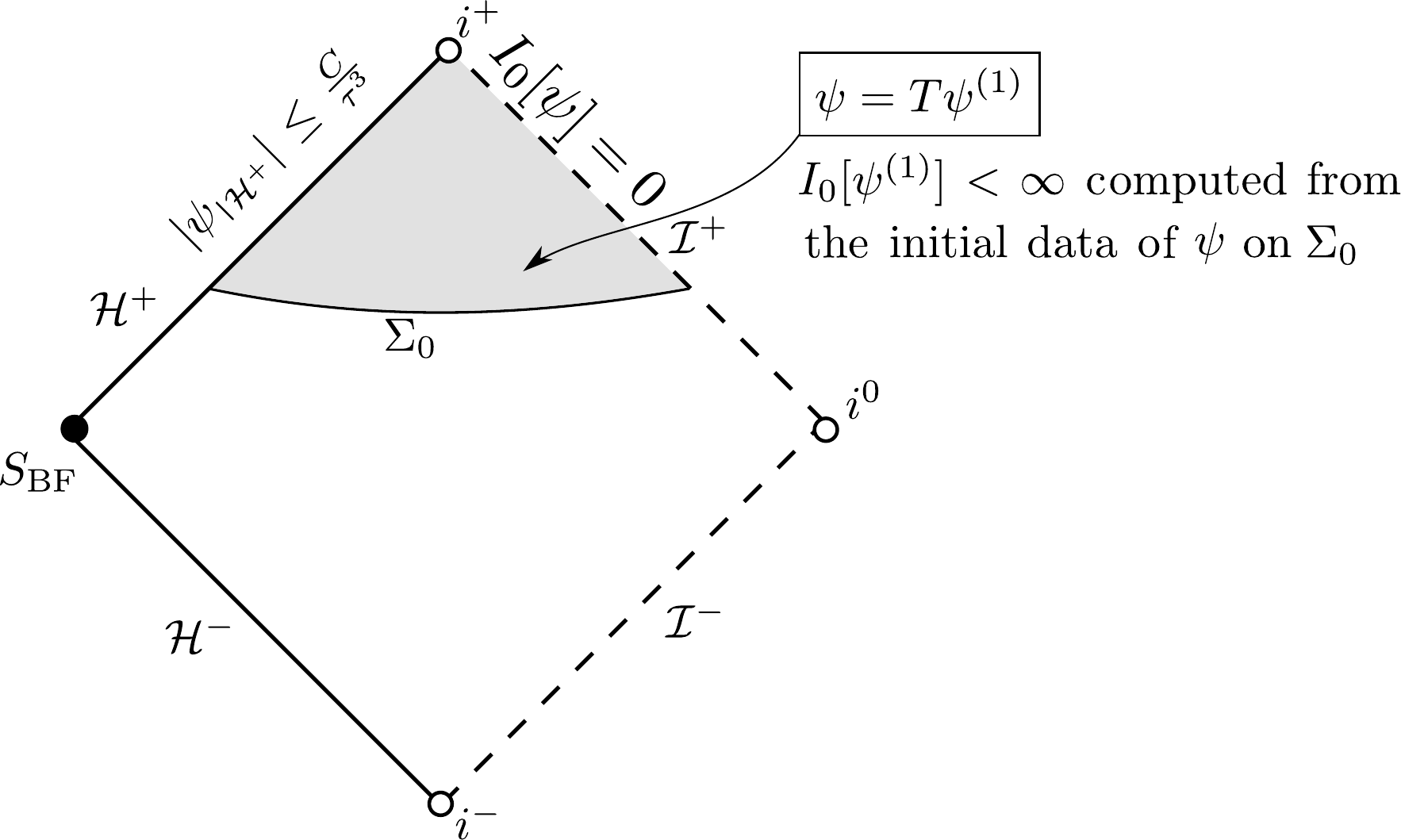}
\caption{\label{fig:65}The domain of $T^{-1}$ has vanishing NP constant and its range has finite NP constant.}
\end{center}
\end{figure}
\vspace{-0.9cm} In particular, if $\Sigma_0$ is  outgoing null for all $r\geq R$, for some large $R>0$, then we have the following formula (\cite{paper2}):
\begin{equation}
\begin{split}
4\pi I_0[\psi^{(1)}]=& -\lim_{r_0\to \infty}\int_{\Sigma_0\cap \{r=r_0\}}r^3\partial_{r}(r\psi) \, d\omega+M\int_{\Sigma_0\cap \{r=R\}} r^2(2-Dh_{\Sigma_0})\psi\, d\omega\\&+M\int_{\Sigma_0\cap\{r\geq R\}}rD \partial_r(r\psi)\,dr'd\omega\\
&-M\int_{\Sigma_0\cap\{r\leq R\}} \Big(2(1-h_{\Sigma_0}D)r\partial_{\rho}(r\psi)-(2-Dh_{\Sigma_0})r^2h_{\Sigma_0} T\psi\\
&-(r^2\cdot (Dh_{\Sigma_0})')\cdot\psi\Big)\,d\rho'd\omega,
\end{split}
\label{generali}
\end{equation}
where $D=1-\frac{2M}{r}$ for Schwarzschild spacetimes, $\partial_r=\frac{2}{D}\partial_v$, $\rho:=r|_{\Sigma_0}$ and $\partial_{\rho}$ is the radial derivative tangential to $\Sigma_0$ that is taken with respect to the induced coordinate system $(\rho,\omega)$ in $\Sigma_0$, and finally $h_{\Sigma_0}$ is defined by the equation
\[\partial_{\rho}= -2D^{-1}\partial_u+h_{\Sigma_0}T. \]
For example, $h_{\{t=0\}}=\frac{1}{D},$ $h_{\{v=v_0 \}}=0,$ $h_{\{u=u_0 \}}=\frac{2}{D}$.

Note that if the initial data for $\psi$ is compactly supported in $\{ r\leq R\}$  then \eqref{generali} reduces to 
\begin{equation}
\begin{split}
&4\pi I_0[\psi^{(1)}]= \\
-M&\int_{\Sigma_0\cap \{r\leq R\}} \Big(2(1-h_{\Sigma_0}D)r\partial_{\rho}(r\psi)-(2-Dh_{\Sigma_0})r^2h_{\Sigma_0} T\psi-(r^2\cdot (Dh_{\Sigma_0})')\cdot\psi\Big)\,d\rho'd\omega,
\end{split}
\label{totakelimit}
\end{equation}
We refer to the constant $I_0[\psi^{(1)}]$ as the \textbf{time-inverted Newman--Penrose (TINP)} constant of $\psi$ and we denote it by $I_{0}^{(1)}[\psi]$.

It is shown in \cite{paper2}  that the coefficient $\boldsymbol{Q}_{\Sigma_0}[\psi]$ in the asymptotic estimates \eqref{eq:ours} and \eqref{eq:oursradiation} is given by 
\begin{equation}
\boldsymbol{Q}_{\Sigma_0}[\psi]=-8I_{0}^{(1)}[\psi].
\label{qi}
\end{equation}
Hence, we obtain the following estimates:
\begin{equation}
\left|\psi(\tau, r_0,\theta,\varphi)+8 I_{0}^{(1)}[\psi]\cdot\frac{1}{\tau^3}\right| \leq C_{r_0} \cdot \sqrt{E_{\Sigma_0}[\psi]}\cdot \frac{1}{\tau^{3+\epsilon}},
\label{eq:ours1}
\end{equation}
\begin{equation}
\left|r\psi|_{\mathcal{I}^{+}}(\tau,\theta,\varphi)+2 I_{0}^{(1)}[\psi]\cdot\frac{1}{\tau^2}\right| \leq C \cdot \sqrt{E_{\Sigma_0}[\psi]}\cdot \frac{1}{\tau^{2+\epsilon}} 
\label{eq:oursradiation1}
\end{equation}

\subsection{Overview of the main results}
\label{outline}

The aim of the present paper is to derive the asymptotic behavior for solutions to the wave equation \eqref{we} with smooth, compactly supported initial data\footnote{More generally, we consider initial data decaying sufficiently fast as $r\rightarrow \infty$.} on Cauchy hypersurfaces \textbf{which pass through the bifurcation sphere}. For simplicity, we will consider in this section smooth, compactly supported initial data on the Cauchy hypersurface $\{t=0\}$.\footnote{Here $t$ is the standard Schwarzschild time coordinate.} Clearly, for such initial data the estimates \eqref{eq:ours1} and \eqref{eq:oursradiation1} hold. Indeed,  solving locally the wave equation \eqref{we} from $\{t=0\}$ to a Cauchy hypersurface $\Sigma_0$ in the future of $\{t=0\}$, which moreover does not intersect the bifurcation sphere, gives rise to smooth initial data on $\Sigma_0$ with vanishing Newman--Penrose constant.\footnote{Note that the induced data on $\Sigma_0$ will not be compactly supported unless $\Sigma_0$ is contained in the domain of dependence of the region $\{r\geq R\}\cap \{t=0\}$ where the solution is zero. Nonetheless, in view of the conservation law discussed in Section \ref{sec:ReviewOfTheTimeIntegralConstruction}, the Newman--Penrose constant for the induced data on $\Sigma_0$ is necessarily zero.} In fact, as was shown in \cite{paper2}, the condition \eqref{frompaper2} holds for the induced data on $\Sigma_0$. Hence, the time integral construction can be applied in the region $\mathcal{J}^{+}(\Sigma_0)$ which allows us to obtain the estimates \eqref{eq:ours1} and \eqref{eq:oursradiation1}, where $I_0^{(1)}[\psi]$ is given by \eqref{generali} in terms of the induced data on $\Sigma_0$. Note that in this case,  the estimates \eqref{eq:ours1} and \eqref{eq:oursradiation1} in principle \textbf{provide only upper bounds} since we a priori have no control on the constant $\boldsymbol{Q}_{\Sigma_0}[\psi]$ in terms of the initial data on $\{t=0\}$. Indeed, in order for \eqref{eq:ours1} and \eqref{eq:oursradiation1} to yield lower bounds we must prove that initial data on $\{t=0\}$ are consistent with 
\begin{equation}
I_0^{(1)}[\psi]\neq 0.
\label{qnecessary}
\end{equation}
However, such a condition cannot be a priori confirmed for initial data on $\{t=0\}$. A simple backwards construction can be used to show the existence of smooth compactly supported initial data on $\{t=0\}$ for which the condition  \eqref{qnecessary} holds. Hence, it immediately follows that for generic smooth compactly supported initial data on $\{t=0\}$ the condition \eqref{qnecessary} holds. However, the following issue remains unresolved:
\begin{itemize}
	\item Find all initial data on $\{t=0\}$ which satisfy \eqref{qnecessary}. 
\end{itemize}
In fact, we would like to address the following more general issue:
\begin{itemize}
\item Obtain an explicit expression of the coefficient $\boldsymbol{Q}_{\Sigma_0}[\psi]$ in terms of the initial data on $\{t=0\}$. 
\end{itemize}
It is only after we have resolved these issues above that we can for example provide a complete characterization of all initial data on $\{t=0\}$ which produce solutions to \eqref{we} which satisfy Price's inverse polynomial law $\tau^{-3}$ as a \textit{lower} bound. 

Clearly, in view of the vanishing of $T$ on the bifurcation sphere,  the time integral construction of Section \ref{sec:ReviewOfTheTimeIntegralConstruction} breaks down in the region $\mathcal{J}^{+}(\{t=0\})$ and hence cannot be used to express the coefficient $\boldsymbol{Q}_{\Sigma_0}[\psi]$ as a time-inverted Newman--Penrose constant. Similarly, we cannot simply evaluate the right hand side of \eqref{generali} on $\{t=0\}$. One legitimate approach would be to consider a sequence of Cauchy hypersurfaces $\Sigma_{i}$ such that 
\begin{enumerate}
	\item for all $i$, $\Sigma_{i}$ intersects the event horizon to the future of the bifurcation sphere,
	\item as $i\rightarrow \infty$ the hypersurfaces $\Sigma_i\cap\{r\leq R\}$ tend to the hypersurface $\{t=0\}$.  
\end{enumerate}
Then clearly, for each (finite) $i\geq 0$, we can express $\boldsymbol{Q}_{\Sigma_i}[\psi]$ via the induced data on $\Sigma_i$. We then simply have to examine if the limit $\lim_{i\rightarrow\infty}\boldsymbol{Q}_{\Sigma_i}[\psi]$ is well-defined and subsequently compute it. Although the above procedure is possible, we pursue in this paper a different approach which yields much more general results about the domain of validity, the regularity and the explicit expression via geometric currents of the time-inverted Newman--Penrose constants.

The main new observation is that for any Cauchy hypersurface $\Sigma$ which intersects the event horizon to the future of the bifurcation sphere, the constant $\boldsymbol{Q}_{\Sigma_0}[\psi]$\footnote{This constant is defined via \eqref{generali}, as an integer multiple of the time-inverted Newman--Penrose constant of the time integral in $\mathcal{J}^{+}(\Sigma_0)$.} is given by an appropriate modification of the gradient flux on $\Sigma$:
\[\int_{\Sigma}\nabla\psi\cdot {n}_{\Sigma}\, d\mu_{\Sigma},\]
where ${n}_{\Sigma}$ is the normal to $\Sigma$ and the integral is taken with respect to the standard volume form $d\mu_{\Sigma}$ corresponding to the induced metric on $\Sigma$. 
Note that the above gradient flux is generically infinite, however, the following modified flux 
\begin{equation}
\lim_{r_0\rightarrow \infty}\left(\int_{\Sigma\cap\{r\leq r_0\}} \nabla\psi\cdot {n}_{\Sigma}\,d\mu_{\Sigma}+\int_{\Sigma\cap\{r=r_0\}}\Big(\psi-\frac{2}{M}r\partial_v(r\psi)\Big)r^2d\omega\right)
\label{trunc1}
\end{equation}
is indeed finite for all hypersurfaces $\Sigma$ (see Lemma \ref{remarkderivation}), where $\partial_v$ is the standard outgoing null derivative. If we define
\begin{equation}
G(\Sigma^{\leq r_0})[\psi]=\int_{\Sigma\cap\mathcal{H}^{+}}\!\!\psi\,r^2d\omega+\int_{\Sigma\cap\{r\leq r_0\}} n_{\Sigma}(\psi)\,d\mu_{\Sigma}+\int_{\Sigma\cap\{r=r_0\}}\Big(\psi-\frac{2}{M}r\partial_v(r\psi)\Big)\,r^2d\omega,
\label{trunc1b}
\end{equation}
then the main new result of this paper is the following identity for the TINP constant of $\psi$ 
\begin{equation}
I_{0}^{(1)}[\psi]=\frac{M}{4\pi}\lim_{r_0\rightarrow \infty}G(\Sigma^{\leq r_0})[\psi].
\label{geomiequa1}
\end{equation}
\begin{figure}[H]
\begin{center}
\includegraphics[width=11cm]{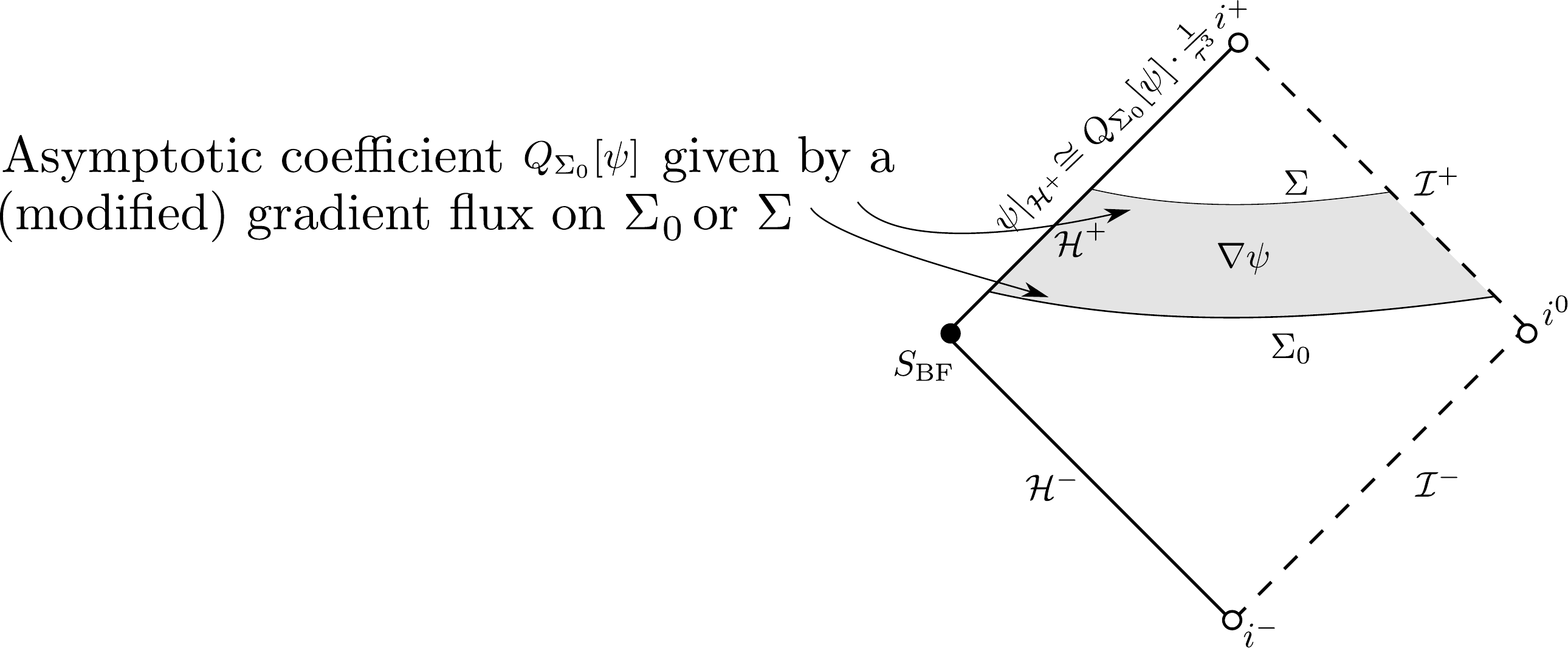}
\caption{\label{fig:6}The asymptotic coefficient as a gradient flux through Cauchy hypersurfaces.}
\end{center}
\end{figure}
\vspace{-0.9cm}
This provides a new geometric interpretation of the coefficient $\boldsymbol{Q}_{\Sigma_0}[\psi]$ of the leading-order terms in the asymptotic expansion (which, recall, is equal to $-8I_{0}^{(1)}[\psi]$) as an appropriately modified gradient flux. Clearly, by the definition of the wave equation, the gradient of a scalar field solution is divergence-free and hence satisfies conservation laws in any compact region. It turns out that for solutions to \eqref{we} with vanishing Newman--Penrose constant, \textit{the modified gradient flux given by the right hand side of \eqref{geomiequa1} satisfies a conservation law for all (unbounded) regions bounded by Cauchy hypersurfaces}. In other words, the limit $\lim_{r_0\rightarrow \infty}G(\Sigma^{\leq r_0})[\psi]$ is independent of the choice of hypersurface $\Sigma$ (see Proposition \ref{prop2deri}).  This conservation law immediately allows us to compute the value of 
$I_{0}^{(1)}[\psi]$ in terms of the initial data on hypersurfaces passing through the bifurcation sphere, even though the former was originally defined in terms of the time integral construction in the region $\mathcal{J}^{+}(\Sigma_0)$ which does not contain the bifurcate sphere and hence \textit{formally extend the domain of validity of $I_{0}^{(1)}[\psi]$ to all Cauchy hypersurfaces regardless of the vanishing of $T$.} For example, for smooth, compactly supported initial data on the hypersurface $\{t=0\}$ we have
 \begin{equation}\boxed{
I_0^{(1)}[\psi]= \frac{M}{4\pi}\int_{ S_{\text{BF}}}\!\!\psi \, r^2d\omega+\frac{M}{4\pi}\int_{\{t=0\}}\ \frac{1}{1-\frac{2M}{r}}\partial_t\psi\, r^2 dr d\omega,}
 \label{ont0}
 \end{equation}
where $S_{\text{BF}}$ denotes the bifurcation sphere $\{t=0\}\cap \{r=2M\}$. 

Note that the second integral on the right hand side is finite since $T=\partial_t$ vanishes at $\{t=0\}\cap \{r=2m\}$ and in fact satisfies
\begin{equation}
n_{\{t=0\}}=\frac{1}{\sqrt{1-\frac{2M}{r}}} \cdot \partial_t,
\label{nt0}
\end{equation}
where $n_{\{t=0\}}$ is the unit normal to $\{t=0\}$.    Hence, \eqref{ont0} allows us to explicitly compute the coefficient in the asymptotic estimates \eqref{eq:ours1} and \eqref{eq:oursradiation1}  in terms of the initial data on $\{t=0\}$.\footnote{ It is important to remark that had we simply evaluated the expression for $I_0^{(1)}[\psi]$ using \eqref{generali} on $t=0$ (for which $h_{\{t=0\}}=\frac{1}{1-\frac{2M}{r}}$)  then we would have missed the first term on the right hand side of \eqref{ont0}.}
\begin{figure}[H]
\begin{center}
\includegraphics[width=14cm]{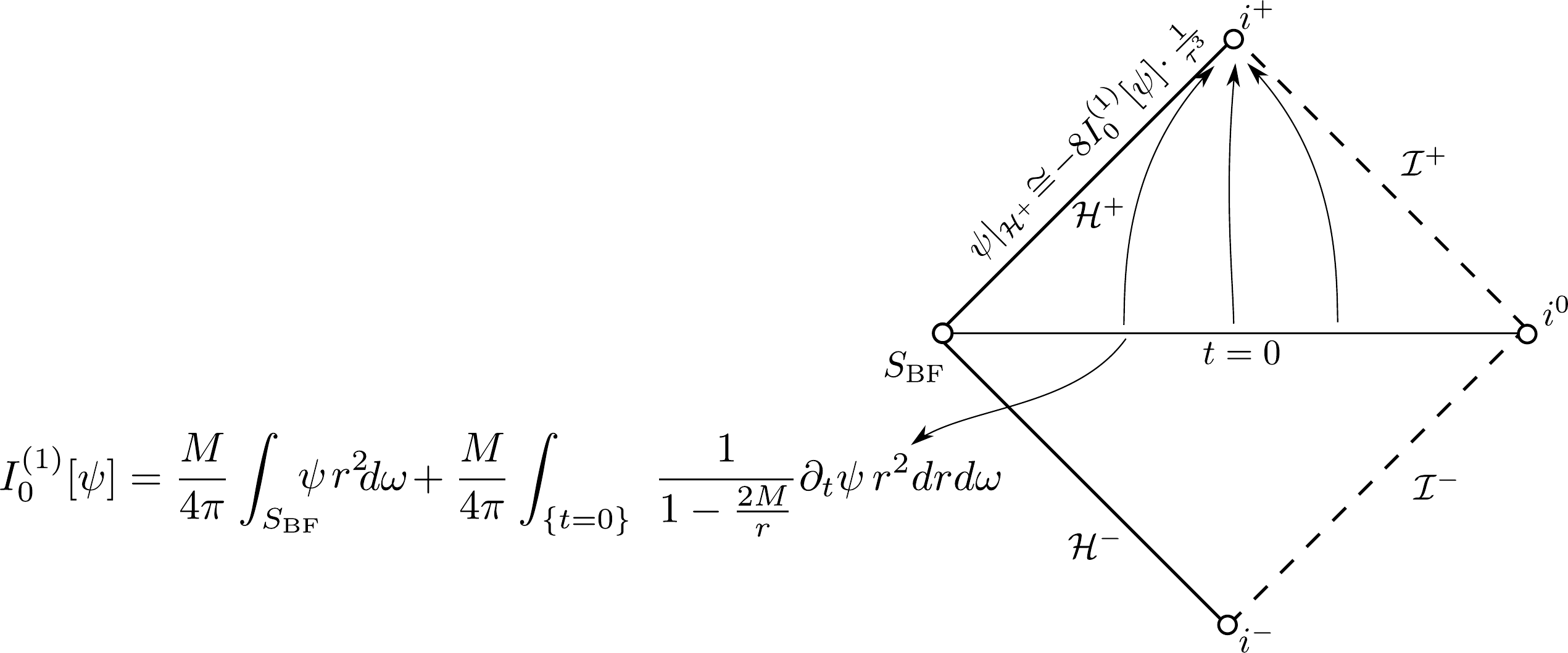}
\caption{\label{fig:7}The asymptotic coefficient in terms of the initial data on $\{t=0\}$.}
\end{center}
\end{figure}
\vspace{-0.9cm}
The new formula \eqref{ont0} allows us to reach several interesting conclusions.
\begin{enumerate}
\item \textit{Evolution/Fall-off of general data:} The late-time asymptotics for smooth compactly supported initial data on $\{t=0\}$ is given by \eqref{eq:ours1}, \eqref{eq:oursradiation1} and \eqref{ont0}. Schematically, we have asymptotically in time as $\tau\rightarrow \infty$
\[\psi(\tau, r_0,\theta,\varphi)\sim -8I_0^{(1)}[\psi]\cdot \frac{1}{\tau^3}, \ \ \ r\psi|_{\mathcal{I}^{+}}(\tau, r=\infty,\theta,\varphi)\sim  -2I_{0}^{(1)}[\psi]\cdot \frac{1}{\tau^2},\]
where $I_0^{(0)}[\psi]$ is given by the explicit expression \eqref{ont0} of the initial data on  $\{t=0 \}$.
	\item \textit{Evolution/Fall-off of time-symmetric (initially static) data:} Initial data on $\{t=0\}$ are called static (or time-symmetric) if 
	\begin{equation}
	\partial_t\psi|_{\{t=0\}}=0.
	\label{staticdata}
	\end{equation}
	In this case, the constant $I_{0}^{(1)}[\psi]$ reduces to 
	\[ I_{0}^{(1)}[\psi]= \frac{M}{4\pi}\int_{\{t=0\}\cap S_{\text{BF}}}\!\!\psi \, r^2d\omega.\]
	Hence, \textbf{the restriction to time-symmetric initial data will generically not improve the corresponding fall-off in time.} The fall-off improves by one power under this restriction if and only if, \underline{in addition}, the spherical mean of $\psi$ on the bifurcation sphere vanishes.
	\begin{figure}[H]
\begin{center}
\includegraphics[width=13.5cm]{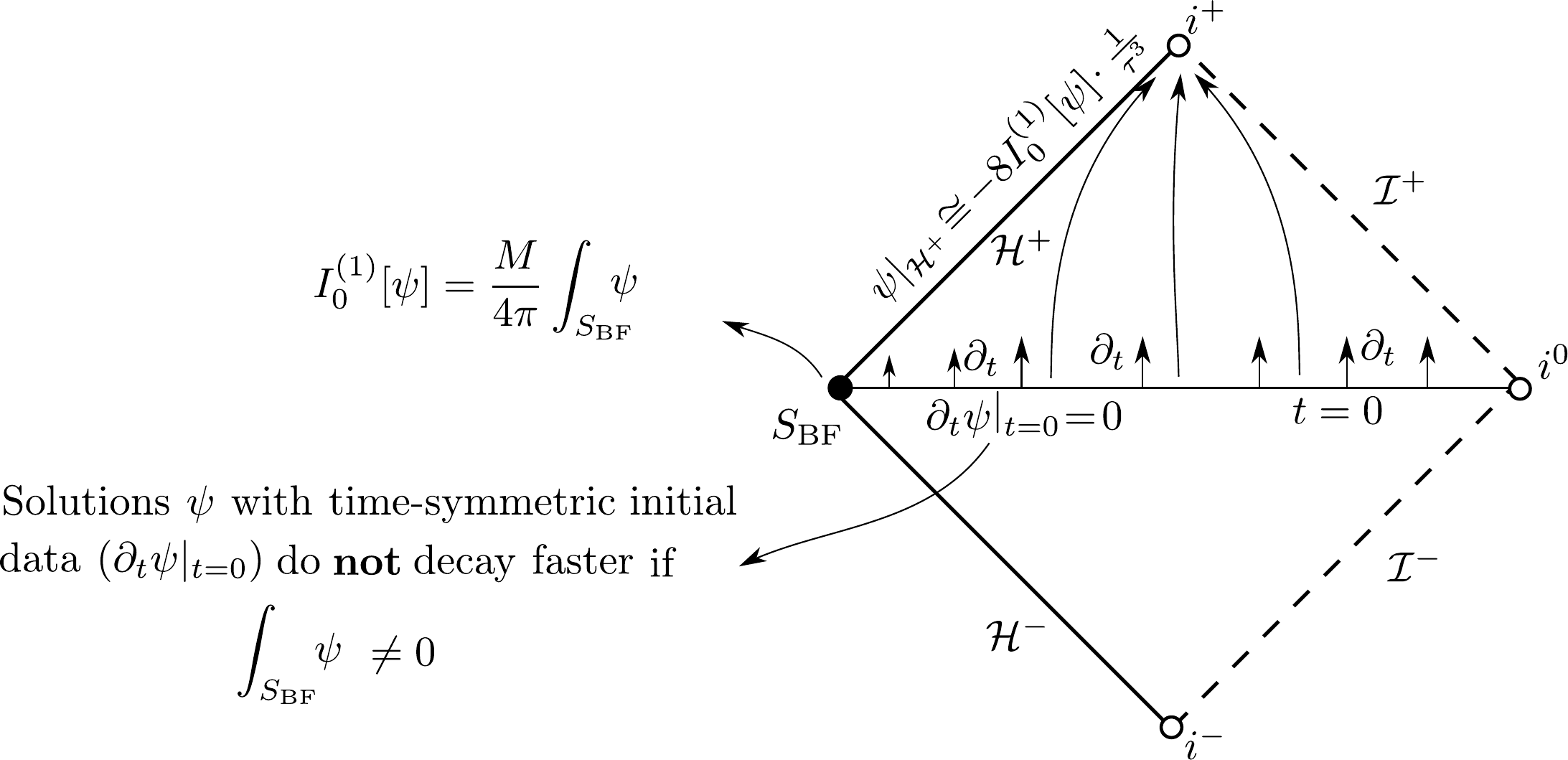}
\caption{\label{fig:7b}The asymptotic coefficient in terms of the initial data on $\{t=0\}$.}
\end{center}
\end{figure}
\vspace{-0.9cm}This provides a complete description of the asymptotic behavior of time-symmetric initial data and hence confirms and extends the numerical work of \cite{karzdma}, the heuristic work of \cite{priceburko}. 

We moreover note that the above observations are consistent with the work of \cite{dssprice} where the authors consider suitably decaying initial data supported away from the bifurcation sphere and show that for each spherical harmonic mode $\psi_{\ell}$ one can estimate
\begin{equation*}
\left| (1+r_*^2)^{-\ell-\frac{3}{2}}\psi_{\ell}\right|(t,r,\theta,\varphi)\lesssim t^{-2-2\ell}D_1[\partial_t\psi_{\ell}|_{t=0}]+ t^{-3-2\ell}D_2[\psi_{\ell}|_{t=0}],
\end{equation*}
where $D_1[\partial_t\psi_{\ell}|_{t=0}]$ is a weighted $L^1$ norm depending only on $\partial_t\psi_{\ell}|_{t=0}$ and $D_2[\psi_{\ell}|_{t=0}]$ is a weighted $L^1$ norm depending only on $\psi_{\ell}|_{t=0}$. Although the decay rates appearing in the estimate above are not the expected sharp decay rates (cf. \eqref{eq:ours} for the $\ell=0$ case and the $\ell$-dependent decay rates suggested by \cite{Price1972}), the estimate illustrates nicely how the decay rate increases by one power if one restricts to initially static data ($\partial_t\psi_{\ell}|_{t=0}=0$).
	
	\item \textit{Evolution/Fall-off of initially vanishing data:} Such data satisfy
	\[\psi|_{\{t=0\}}=0.\]
	In this case $I_{0}^{(1)}[\psi]$ reduces to 
	\[ I_{0}^{(1)}[\psi]= \frac{M}{4\pi}\int_{\{t=0\}}\ \frac{1}{1-\frac{2M}{r}}\partial_t\psi\, r^2 dr d\omega.\]
	Hence generic such initial data have the same fall-off as the general data.
	
\end{enumerate}

Furthermore, the results of \cite{paper2} and the expression \eqref{ont0} of the TINP constant $I_0^{(1)}[\psi]$ allow us to compare the asymptotic behavior of scalar fields based on the profile of the initial data on the following two (types of) Cauchy hypersurfaces:
\begin{enumerate}
	\item $\Sigma_0$ which does not intersect the bifurcation sphere (and hence intersects the event horizon to the future of the bifurcation sphere)
	\item $\{t=0\}$ which passes through the bifurcation sphere. 
\end{enumerate}
For a Cauchy hypersurface $\Sigma$ (of either type above) we consider the following function spaces: 
\begin{equation*}
\begin{split}
H_{I_0<\infty}(\Sigma)=&\left\{\psi \in C^{\infty}\left(\mathcal{J}^{+}(\Sigma)\right): \ \Box_g\psi=0 \text{ and } I_0[\psi]<\infty  \right\},\\
H_{I_0\neq 0}(\Sigma)=&\left\{\psi\in C^{\infty}\left(\mathcal{J}^{+}(\Sigma)\right): \ \Box_g\psi=0 \text{ and } 0\neq I_0[\psi]<\infty  \right\},\\
H_{I_0=0}(\Sigma)=&\left\{\psi\in C^{\infty}\left(\mathcal{J}^{+}(\Sigma)\right): \ \Box_g\psi=0 \text{ with  c.s. and s.s. initial data on } \Sigma  \right\},\\
H_{3}(\Sigma)=&\left\{ \psi\in H_{I_0=0}(\Sigma):\ \tau^3\cdot  \psi(\tau,r_0,\cdot) \nrightarrow 0 \right\},\\
H_{\geq 4}(\Sigma)=&\left\{ \psi\in H_{I_0=0}(\Sigma):\ \tau^3\cdot  \psi(\tau,r_0,\cdot) \rightarrow 0\right\},\\
\end{split}
\end{equation*}
where ``c.s=compactly supported'' and ``s.s.=spherically symmetric''\footnote{ Recall from Section \ref{sec:ReviewOfTheTimeIntegralConstruction} that the spherical mean dominates the asymptotic fall-off behavior.}   and $I_0[\psi]$ denotes the Newman--Penrose constant of $\psi$ at null infinity. Clearly, the Newman--Penrose constant of solutions in the space $H_{I_0=0}(\Sigma)$ vanishes. 

According to the results in \cite{paper2} for the hypersurface $\Sigma_0$ intersecting $\mathcal{I}^+$:
\[ H_{I_0<\infty}(\Sigma_0)\supseteq H_{I_0\neq 0}(\Sigma_0) \supseteq H_{I_0=0}(\Sigma_0)=  H_{3}(\Sigma_0)\cup H_{\geq 4}(\Sigma_0).  \]
Furthermore, \textbf{we have  the following invertibility properties for the operator $T$ in $\mathcal{J}^{+}(\Sigma_0)$}:
\begin{equation}
\text{ For all }\psi\in H_{I_0=0}(\Sigma_0) \text{ there is } \psi^{(1)}\in H_{I_0<\infty}(\Sigma_0)  \text{ such that } T\psi^{(1)}=\psi. 
\label{statement1}
\end{equation}Furthermore, $\psi(\tau,r_0,\cdot)$ decays \textbf{at least as fast as }$\tau^{-3}$ and $\psi^{(1)}(\tau,r_0,\cdot)$ decays at least as fast as $\tau^{-2}$.  On the other hand, we have
\begin{equation}
\text{ For all }\psi\in H_{3}(\Sigma_0) \text{ there is } \psi^{(1)}\in H_{I_0\neq 0}(\Sigma_0)  \text{ such that } T\psi^{(1)}=\psi.
\label{statement2}
\end{equation}
The above characterizes all solutions $\psi$ which decay \textbf{exactly} like $\tau^{-3}$. The following characterizes all solutions which decay {at least} one power faster, that is \textbf{at least as fast as} $\tau^{-4}$:
\begin{equation}
\text{ For all }\psi\in H_{\geq 4}(\Sigma_0) \text{ there is } \psi^{(1)}\in H_{I_0<\infty}(\Sigma_0)  \text{ such that } TT\psi^{(1)}=\psi.
\label{statement3}
\end{equation}

\textbf{The $T$-invertibility statements \eqref{statement1}, \eqref{statement2} and  \eqref{statement3} are not valid for the hypersurface $\{t=0\}$}\footnote{The $T$-invertibility properties had already been studied by Wald \cite{drimos} and Kay--Wald \cite{Kay1987} in the context of obtaining uniform boundedness for solutions to the wave equation.}. 
\begin{itemize}
	\item Statement \eqref{statement1} is clearly not true for the hypersurface $\{t=0\}$ since   for $\psi\in H_{I_0=0}(\{t=0\})$ we generically have $\psi|_{S_{\text{BF}}}\neq 0$ whereas $T=0$ at the bifurcation sphere. 
		\begin{figure}[H]
\begin{center}
\includegraphics[width=12cm]{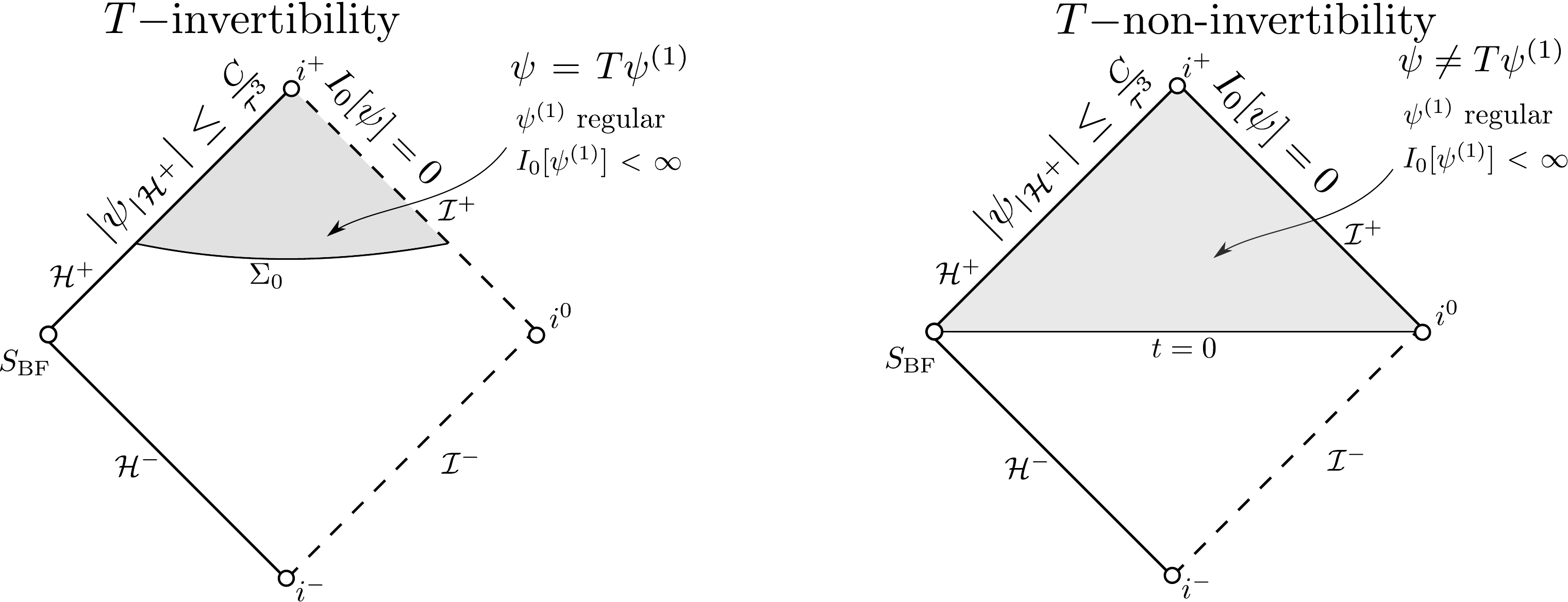}
\caption{\label{fig:10}$T-$invertibility in $\mathcal{J}^{+}(\Sigma_0)$ and $T-$non-invertibility in $\mathcal{J}^{+}(\{t=0\})$.}
\end{center}
\end{figure}
\vspace{-0.9cm}
	\item Statement \eqref{statement2} is not true for $\{t=0\}$ since all  $\psi\in H_{3}(\{t=0\})$ satisfy, in view of \eqref{ont0},
	\[\frac{M}{4\pi}\int_{ S_{\text{BF}}}\!\!\psi \, r^2d\omega+\frac{M}{4\pi}\int_{\{t=0\}}\ \frac{1}{1-\frac{2M}{r}}\partial_t\psi\, r^2 dr d\omega\neq 0,\] and hence, they generically satisfy $\psi|_{S_{\text{BF}}}\neq 0$, whereas $T=0$ at the bifurcation sphere.
	\begin{figure}[H]
\begin{center}
\includegraphics[width=12cm]{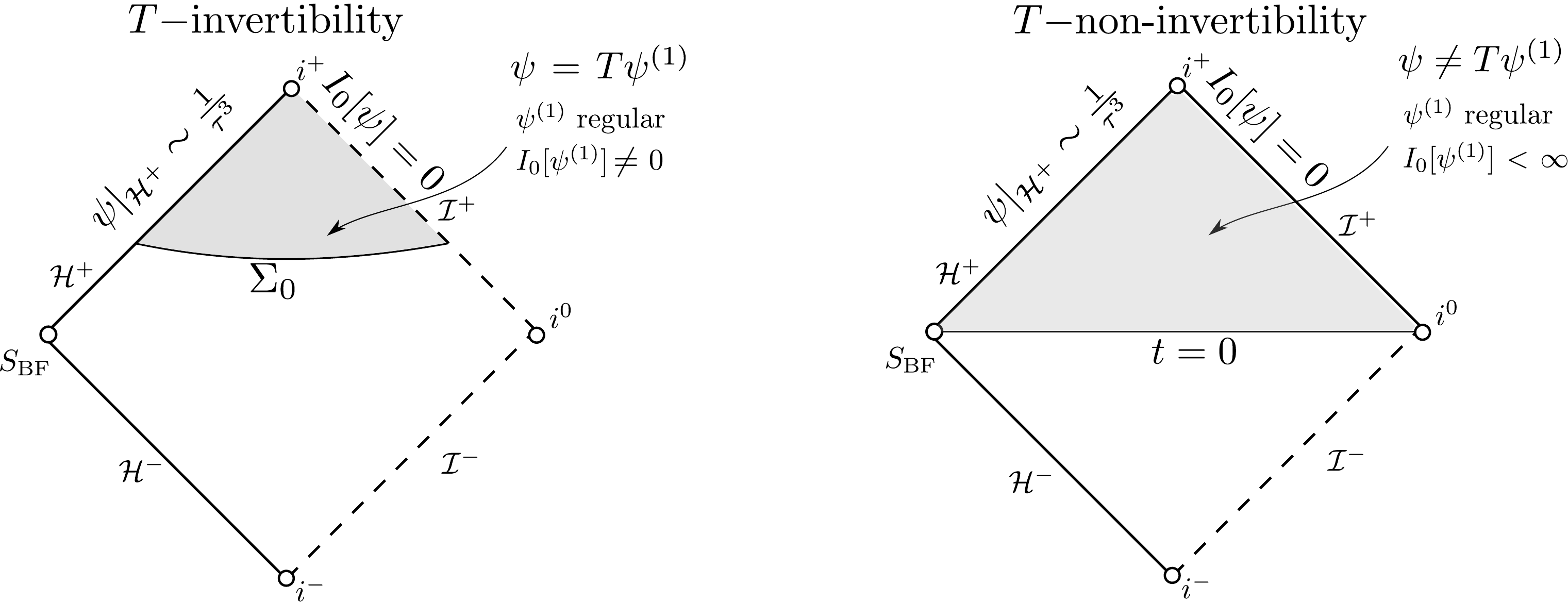}
\caption{\label{fig:11}$T-$invertibility in $\mathcal{J}^{+}(\Sigma_0)$ and $T-$non-invertibility in $\mathcal{J}^{+}(\{t=0\})$.}
\end{center}
\end{figure}
\vspace{-0.9cm}
	\item  Statement \eqref{statement3} is not true for $\{t=0\}$ since all  $\psi\in H_{\geq 4}(\{t=0\})$ satisfy, in view of \eqref{ont0},
	\[\frac{M}{4\pi}\int_{ S_{\text{BF}}}\!\!\psi \, r^2d\omega+\frac{M}{4\pi}\int_{\{t=0\}}\ \frac{1}{1-\frac{2M}{r}}\partial_t\psi\, r^2 dr d\omega= 0,\] and hence, they generically still satisfy $\psi|_{S_{\text{BF}}}\neq 0$, whereas $T=0$ at the bifurcation sphere. 
	\begin{figure}[H]
\begin{center}
\includegraphics[width=12cm]{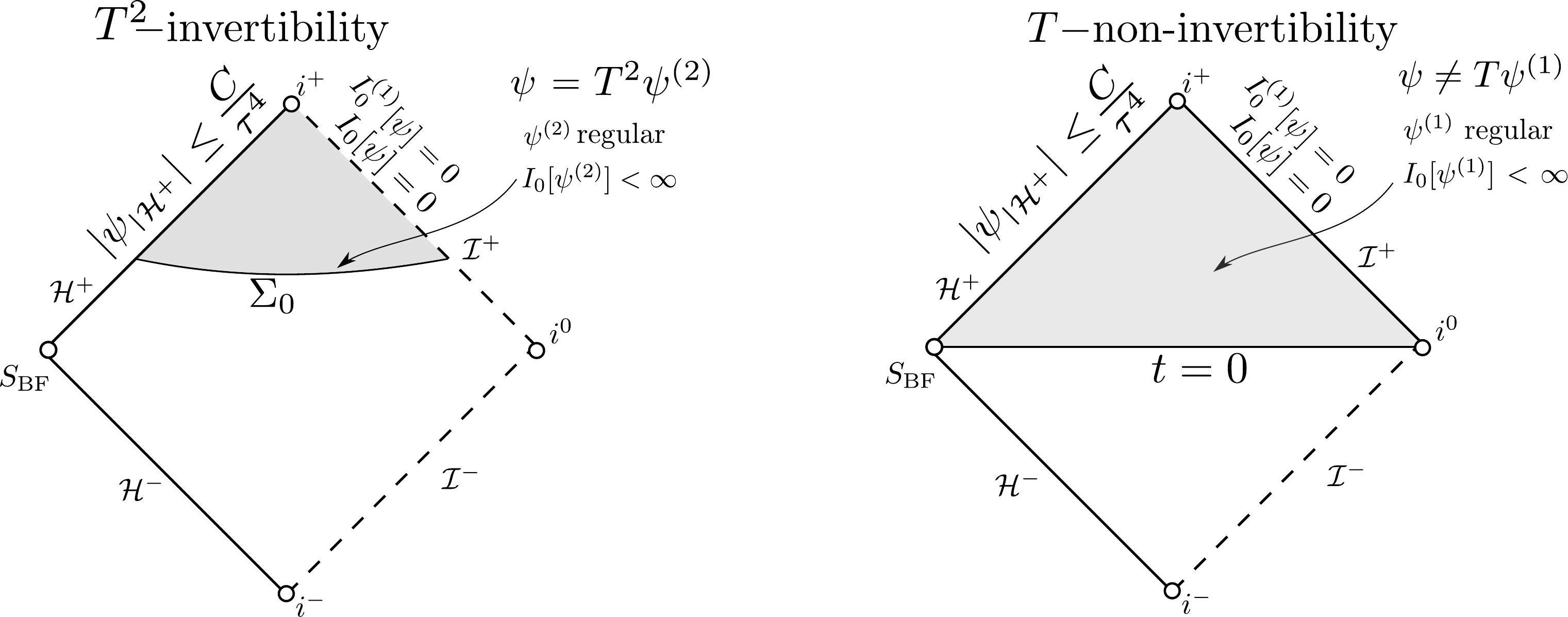}
\caption{\label{fig:12}$T^2-$invertibility in $\mathcal{J}^{+}(\Sigma_0)$ and $T-$non-invertibility in $\mathcal{J}^{+}(\{t=0\})$.}
\end{center}
\end{figure}
\vspace{-0.9cm}
\end{itemize}
In short, we conclude that  smooth solutions to \eqref{we} in $\mathcal{J}^{+}(t=0)$ that decay strictly faster than $\tau^{-3}$ generically 
\begin{enumerate}
	\item do not arise from time-symmetric initial data on $\{t=0\}$, and
	\item do not arise as the $T$ derivative of regular solutions to the wave equation \eqref{we} in $\mathcal{J}^{+}(\{t=0\})$.
\end{enumerate}

\subsection{Relation to scattering theory}
\label{discussion}

An additional convenience of the fact that we can ``read off'' the late-time asymptotics from the initial data on $\{t=0\}$ is that we can evolve such data both to the future and to the past and hence obtain a correlation between $(\psi|_{\mathcal{H}^{-}},r\psi|_{\mathcal{I}^{-}})$, the induced $\psi$ and $r\psi$ on the past event horizon $\mathcal{H}^{-}$ and the past null infinity $\mathcal{I}^{-}$, respectively, and $(\psi|_{\mathcal{H}^{+}},r\psi|_{\mathcal{I}^{+}})$, the induced $\psi$ and $r\psi$ on the future event horizon $\mathcal{H}^{+}$ and the future null infinity $\mathcal{I}^{+}$, respectively. For convenience, we will restrict the discussion in this section to smooth compactly supported initial data $\left(\psi|_{\{t=0\}},\partial_t\psi_{\{t=0\}}\right)$ on $\{t=0\}$.

It is in fact possible to consider more generally the evolution of ``past scattering data'' $(\psi|_{\mathcal{H}^{-}},r\psi|_{\mathcal{I}^{-}})$ to ``future scattering data'' $(\psi|_{\mathcal{H}^{+}},r\psi|_{\mathcal{I}^{+}})$ and vice versa via a \emph{scattering map}, a bijection between suitable energy spaces on $\mathcal{I}^-\cup \mathcal{H}^-$ and $\mathcal{I}^+\cup \mathcal{H}^+$. We refer to \cite{DafShl2016,linearscattering} and the references therein for results pertaining to the scattering map. Note however that \emph{the evolution of such scattering data need not result in a solution $\psi$ with $\left(\psi|_{\{t=0\}},\partial_t\psi_{\{t=0\}}\right)$ smooth and compactly supported}. By imposing smoothness and compact support of $\left(\psi|_{\{t=0\}},\partial_t\psi_{\{t=0\}}\right)$, we are therefore restricting to \emph{special} scattering data from the point of view of the scattering map.

For convenience let us denote by 
\[I_0^{(1)}[\psi,\mathcal{I}^{+}], \  \ I_0^{(1)}[\psi,\mathcal{I}^{-}]\]
the TINP constants of $\psi$ on the future and past null infinity $\mathcal{I}^{+}, \mathcal{I}^{-}$, respectively.\footnote{So far we have only considered the future region and hence $I_{0}^{(1)}[\psi]$ has always been the TINP constant on $\mathcal{I}^{+}$.}

If we restrict to smooth compactly supported initial data $\left(\psi|_{\{t=0\}},\partial_t\psi_{\{t=0\}}\right)$ on $\{t=0\}$ then, in view of \eqref{ont0}, the coefficient of the leading-order future-asymptotic term\footnote{See, for instance, \eqref{eq:ours} and \eqref{eq:oursradiation}.} along both $\mathcal{H}^{+}$ and $\mathcal{I}^{+}$ is given by 
\[I_0^{(1)}[\psi,\mathcal{I}^{+}]=\frac{M}{4\pi}\int_{ S_{\text{BF}}}\!\!\psi \, r^2d\omega+\frac{M}{4\pi}\int_{\{t=0\}}\ \frac{1}{1-\frac{2M}{r}}\partial_t\psi\, r^2 dr d\omega.\]Hence, asymptotically along $\mathcal{I}^{+}$ as $\tau\rightarrow \infty$ (towards past timelike infinity)
\[r\psi|_{\mathcal{I}^{+}}(\tau, r=\infty,\theta,\varphi)\sim  -2I_{0}^{(1)}[\psi,\mathcal{I}^{+}]\cdot \frac{1}{\tau^2}.\]
Similarly, in view of the time symmetry of the Schwarzschild metric, we obtain that
 the coefficient of the leading-order past-asymptotic term along both $\mathcal{H}^{-}$ and $\mathcal{I}^{-}$ is given by 
\[I_0^{(1)}[\psi,\mathcal{I}^{-}]=\frac{M}{4\pi}\int_{ S_{\text{BF}}}\!\!\psi \,r^2 d\omega-\frac{M}{4\pi}\int_{\{t=0\}}\ \frac{1}{1-\frac{2M}{r}}\partial_t\psi\, r^2 dr d\omega.\]
Then, asymptotically along $\mathcal{I}^{-}$ as $\tau\rightarrow \infty$
\[r\psi|_{\mathcal{I}^{-}}(\tau, r=\infty,\theta,\varphi)\sim  -2I_{0}^{(1)}[\psi,\mathcal{I}^{-}]\cdot \frac{1}{\tau^2}.\]
Hence, we obtain
\begin{equation}
I_0^{(1)}[\psi,\mathcal{I}^{+}]=-I_0^{(1)}[\psi,\mathcal{I}^{-}]+\frac{M}{2\pi}\int_{ S_{\text{BF}}}\!\!\psi \, r^2d\omega.
\label{qfuturepast}
\end{equation}
	\begin{figure}[H]
\begin{center}
\includegraphics[width=7cm]{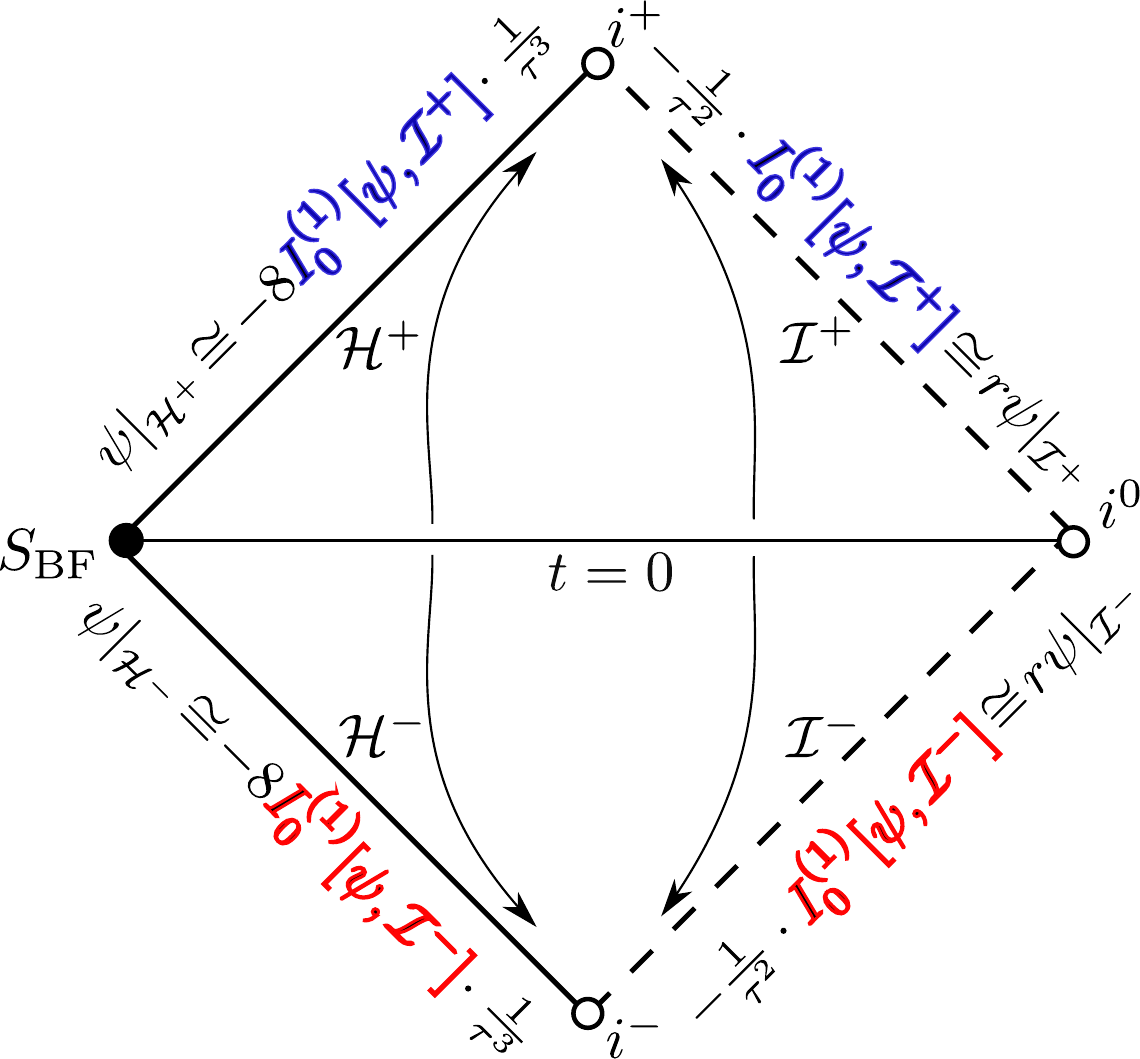}
\caption{\label{fig:9}The scattering map and the future and past TINP constants $I_{0}^{(1)}[\psi,\mathcal{I}^{+}]$, $I_{0}^{(1)}[\psi,\mathcal{I}^{-}]$.}
\end{center}
\end{figure}
\vspace{-0.9cm}
Hence, for this special class of solutions to the wave equation arising from past scattering data $(\psi|_{\mathcal{H}^{-}},r\psi|_{\mathcal{I}^{-}})$ \textbf{such that $\psi|_{t=0}$ and $n_{t=0}\psi|_{t=0}$ are smooth and compactly supported}, the (leading order) asymptotic behavior of $(\psi|_{\mathcal{H}^{+}},r\psi|_{\mathcal{I}^{+}})$ is determined by 
\begin{enumerate}
	\item the (leading order) asymptotic behavior of the past scattering data $(\psi|_{\mathcal{H}^{-}},r\psi|_{\mathcal{I}^{-}})$, and
	\item the spherical mean of $\psi$ on the bifurcation sphere $S_{\text{BF}}=\mathcal{H}^{+}\cap\mathcal{H}^{-}$. 
\end{enumerate}

We can moreover go beyond the leading order asymptotics and obtain a direct integral relation between the \emph{full} scattering data sets $(\psi|_{\mathcal{H}^{+}},r\psi|_{\mathcal{I}^{+}})$ and $(\psi|_{\mathcal{H}^{-}},r\psi|_{\mathcal{I}^{-}})$.

First, we note that we can express the coefficients $I_0^{(1)}[\psi,\mathcal{I}^{-}]$ and $I_0^{(1)}[\psi,\mathcal{I}^{+}]$ solely in terms of the past scattering data $(\psi|_{\mathcal{H}^{-}},r\psi|_{\mathcal{I}^{-}})$ or the future scattering data $(\psi|_{\mathcal{H}^{+}},r\psi|_{\mathcal{I}^{+}})$. Indeed, for $\psi$ arising from smooth, compactly supported data $\{t=0\}$, we have that
\begin{align}
\label{TINPfutnullinfinityint}
I_0^{(1)}[\psi,\mathcal{I}^{+}]=&\frac{M}{4\pi}\int_{-\infty}^{\infty} \int_{\s^2 }r\psi|_{\mathcal{I}^+}(u,\theta,\varphi)\,d\omega du,\\
\label{TINPpastnullinfinityint}
I_0^{(1)}[\psi,\mathcal{I}^{-}]=&\frac{M}{4\pi}\int_{-\infty}^{\infty} \int_{\s^2 }r\psi|_{\mathcal{I}^-}(v,\theta,\varphi)\,d\omega dv.
\end{align}
See also Section 1.6 of \cite{paper2}.\footnote{The above integrals moreover play an important role in \cite{luk2015}.}

By combining \eqref{qfuturepast} with \eqref{TINPfutnullinfinityint} and \eqref{TINPpastnullinfinityint}, we obtain the following relation between $(\psi|_{\mathcal{H}^{-}},r\psi|_{\mathcal{I}^{-}})$ and  $(\psi|_{\mathcal{H}^{+}},r\psi|_{\mathcal{I}^{+}})$ (with $\psi$ arising from smooth, compactly supported data $\{t=0\}$):
\begin{equation*}
\int_{-\infty}^{\infty} \int_{\s^2 }r\psi|_{\mathcal{I}^+}(u,\theta,\varphi)\,d\omega du=-\int_{-\infty}^{\infty} \int_{\s^2 }r\psi|_{\mathcal{I}^-}(v,\theta,\varphi)\,d\omega dv+2\int_{ S_{\text{BF}}}\!\!\psi \, r^2d\omega.
\end{equation*}

\subsection{The main theorems}
\label{mainresult}
We consider spherically symmetric black hole spacetimes $(\mathcal{M},g)$ as defined in Section \ref{geosetting} including in particular the Schwarzschild and the more general sub-extremal Reissner--Nordstr\"{o}m  family of black hole spacetimes. In particular, the coordinates $u,v,r,t$ are as defined in Section \ref{geosetting}.

The Newman--Penrose constant and the time-inverted Newman--Penrose (TINP) constant are defined in Section \ref{sec:ReviewOfTheTimeIntegralConstruction}.

 Consider a Cauchy hypersurface $\Sigma$ that crosses the (future or past)  event horizon $\mathcal{H}^{+}$  and terminates at (future or past)  null infinity $\mathcal{I}^{+}$. 
We define the truncated quantity:
\begin{equation}
\boxed{G(\Sigma^{\leq r_0})[\psi]=\int_{\Sigma\cap\mathcal{H}^{+}}\!\!\psi\,r^2d\omega+\int_{\Sigma\cap\{r\leq r_0\}} n_{\Sigma}(\psi)\,d\mu_{\Sigma}+\int_{\Sigma\cap\{r=r_0\}}\Big(\psi-\frac{2}{M}r\partial_v(r\psi)\Big)r^2d\omega}.
\label{trunc}
\end{equation}

The following theorem derives a geometric interpretation of the TINP constant on hyperboloidal slices to the future of the bifurcation sphere in terms of an appropriately modified gradient flux. 
\begin{theorem} \textbf{(The TINP constant as a modified gradient flux)}
Consider a Cauchy hypersurface $\Sigma$ that crosses the future event horizon $\mathcal{H}^{+}$ to the future of the bifurcation sphere and terminates at future null infinity $\mathcal{I}^{+}$. Let $\psi$ be a solution to the wave equation \eqref{we} with vanishing Newman--Penrose constant $I_{0}[\psi]=0$ such that in fact
\[\lim_{r\rightarrow \infty} r^3 \partial_v(r\psi)|_{\Sigma}<\infty,  \]
then the time-inverted Newman--Penrose constant $I_{0}^{(1)}[\psi]$ of $\psi$ is given by 
\begin{equation}
I_{0}^{(1)}[\psi]=\frac{M}{4\pi}\lim_{r_0\rightarrow \infty}G(\Sigma^{\leq r_0})[\psi],
\label{geomiequa}
\end{equation}
where $G(\Sigma^{\leq r_0})[\psi]$ is given by \eqref{trunc}. 
\label{mytheo1}
\end{theorem}
 Theorem \ref{mytheo1} is proved in Section \ref{sec:AGeometricInterpretationOfTheTimeInvertedNewmanPenroseConstant}.

 The original definition of the TINP constant breaks down
for Cauchy hypersurfaces emanating from the bifurcation sphere since the time-
integral construction is singular for smooth initial data with non-trivial support on
the bifurcation sphere. The next theorem derives a generalized conservation law for the modified gradient fluxes which allows us to extend the validity of the TINP constant to more general Cauchy hypersurfaces. 

\begin{theorem} \textbf{(Conservation law for the TINP constant)}
\label{thm:consvlaw}

Consider two arbitrary hypersurfaces $\Sigma_i, i=1,2$ which cross the (future or past) event horizon\footnote{The hypersurfaces $\Sigma_i$ are allowed to intersect the bifurcation sphere.} and terminate at future null infinity.

Let $\psi$ be a solution to the wave equation \eqref{we} with vanishing Newman--Penrose constant $I_{0}[\psi]=0$ such that in fact
\[ \lim_{r\rightarrow \infty} r^3 |\partial_{v}(r\psi)|_{\Sigma_1}<\infty,\ \ \ \ \lim_{r\rightarrow \infty} r^3 |\partial_{v}(r\psi)|_{\Sigma_2}<\infty.  \] 
 Then, 
\begin{equation}
\lim_{r_0\rightarrow \infty}G(\Sigma_1^{\leq r_0})[\psi]=\lim_{r_0\rightarrow \infty}G(\Sigma_2^{\leq r_0})[\psi],
\label{geomiequaconb}
\end{equation}
where $G$ is given by \eqref{trunc}. 
In other words, the expression $\lim_{r_0\rightarrow \infty}G(\Sigma^{\leq r_0})[\psi]$ is \textbf{independent} of the choice of hypersurface $\Sigma$.
\label{mytheo2}
\end{theorem}
Theorem \ref{thm:consvlaw} is proved in Section \ref{sec:ConservationLawForTheModifiedGradientFluxes}.

The final theorem obtains an explicit expression for the TINP constant in terms of the initial data on the $\{t=0\}$ hypersurface (which emanates from the bifurcation sphere). 

\begin{theorem} \textbf{(The TINP constant on $\{t=0\}$)}
Let $\psi$ be a solution to the wave equation \eqref{we} with smooth, compactly supported initial data on $\{t=0\}$. Then the time-inverted Newman--Penrose constant $I_0^{(1)}[\psi]$ of $\psi$ is given by
\[I_0^{(1)}[\psi]= \frac{M}{4\pi}\int_{ S_{\text{BF}}}\!\!\psi \, r^2d\omega+\frac{M}{4\pi}\int_{\{t=0\}}\ \frac{1}{1-\frac{2M}{r}}\partial_t\psi\, r^2 dr d\omega.\]
\label{mytheo3}
\end{theorem}
Theorem \ref{mytheo3} is proved in Section \ref{sec:TheTINPConstantOnT0}.

\subsection{Acknowledgements}
\label{Acknowledgements}

We would like to thank Professor Avy Soffer for suggesting the problem and for several insightful discussions. The second author (S.A) would like to thank the Central China Normal University for the hospitality in the period August 5--16, 2017 where part of this work was completed. The second author (S.A.) acknowledges support through NSF grant DMS-1265538, NSERC grant 502581, an Alfred P. Sloan Fellowship in Mathematics and the Connaught Fellowship 503071.

\section{The geometric setting}
\label{geosetting}

We consider stationary, spherically symmetric and asymptotically flat  spacetimes $(\mathcal{M},g)$ as in Section 2.1 of \cite{paper2}. These include in particular the Schwarzschild family and the larger sub-extremal Reissner--Nordstr\"{o}m family of black holes as special cases. In this section we briefly recall the geometric assumptions on the spacetime metrics and introduce the notation that we use in this paper. 
 
The manifold $\mathcal{M}$ is (partially) covered by appropriate double null coordinates $(u,v,\theta,\varphi)$ with respect to which the metric takes the form 
\begin{equation*}
g=-D(r)dudv+r^2(d\theta^2+\sin^2 \theta d\varphi^2),
\end{equation*}
with $D(r)$ a smooth function such that
\begin{equation}
D(r)=1-\frac{2M}{r}+\frac{d_1}{r^2}+O_3(r^{-2-\beta}),
\label{dbehavior}
\end{equation}
where $d_1\in\mathbb{R}$ and $\beta>0$.\footnote{We use here the standard big O notation to indicate terms that can be uniformly bounded by $r^{-2-\beta}$, and moreover, their $k$-th derivatives, with $k\leq 3$, can be uniformly bounded by $r^{-2-\beta-k}$.} Here $r=r(u,v)$ denotes the area-radius of the spheres $S_{u,v}$ of symmetry. The sub-extremal Reissner-Nordstr\"{o}m is a special case with $D(r)=1-\frac{2M}{r}+\frac{e^2}{r^2}$ with $|e|<M$. We assume that   $D(r_{\text{min}})=0$ and that $\left.\frac{dD(r)}{dr}\right|_{r=r_{\text{min}}}\neq 0$ for some $r_{\text{min}}>0$ and that $D(r)>0$ for $r>r_{\text{min}}$. 

The boundary hypersurface $\{r=r_{\text{min}}\}\cap\{v<\infty \}\cap\{u=\infty\}$ is called the future event horizon and is denoted by $\mathcal{H}^{+}$ whereas the boundary hypersurface $\{r=r_{\text{min}}\}\cap\{u<\infty \}\cap\{v=-\infty\}$ is called the past event horizon and is denoted by $\mathcal{H}^{-}$. The future event horizon and the past event horizon intersect at the bifurcation sphere $S_{\text{BF}}$. 

The null hypersurfaces $C_{u_{0}}=\{u=u_0\}$ terminate in the future (as $r,v\rightarrow +\infty$) at future null infinity  $\mathcal{I}^{+}$. Note also that $u$ is a ``time'' parameter along future null infinity $\mathcal{I}^{+}$ such that $u$ increases towards the future. Similarly, the null hypersurfaces $\underline{C}_{v_{0}}=\{v=v_0\}$ terminate in the future (as $u\rightarrow -\infty$) at past null infinity  $\mathcal{I}^{-}$. Note also that $v$ is a ``time'' parameter along the past null infinity $\mathcal{I}^{-}$ such that $v$ increases towards the future. 

Furthermore, by an appropriate normalization (see, for instance, \cite{paper2}) we can assume that 
\begin{equation}
v-u=2r^*,
\label{uvr}
\end{equation} 
where the function $r^*=r^{*}(r)$ is given by
\[ r^*=R+\int_{R}^rD^{-1}(r')dr'.\]
Here $R>0$ is a sufficiently large but fixed constant. For such coordinates, we define the time function $t$ as follows:
\[v=t+r^*,\ \ u=t-r^*.\]
We also define the (stationary) vector field 
\[T=\partial_u+\partial_v.\]
Note that $T=\partial_t$ with respect to the $(t,r,\theta,\varphi)$ coordinate system. 

Finally, we will consider hypersurfaces $\Sigma$ which terminate at future null infinity. We consider the induced coordinate system $(\rho,\theta,\varphi)$ on $\Sigma$, where $\rho=r|_{\Sigma_0}$ and we define the function $h_{\Sigma}$ on $\Sigma$ such that the tangent vector field $\partial_{\rho}$ to $\Sigma$ satisfies the equation
\begin{align*}
\partial_{\rho}=-2D^{-1}\partial_u+h_{\Sigma}T,
\end{align*}
where $h_{\Sigma}$ is a positive function on $\Sigma$, such that $0\leq 2-h_{\Sigma}(r)D(r)=O(r^{-1-\epsilon})$, for some (arbitrarily small) $\epsilon>0$. It will be convenient to employ also the following alternative form of $\partial_{\rho}$:
\begin{equation*}
\partial_v^{\Sigma}=\frac{1}{h_{\Sigma}}\partial_{\rho}=\partial_v-f\partial_u,
\end{equation*}
with $f(r)=\frac{1}{h_{\Sigma}}(2-h_{\Sigma}D)$. Note that we can analogously define hyperboloidal hypersurfaces terminating at past null infinity. Finally, we will denote with $n_{\Sigma}$ the normal vector field to $\Sigma$ and with $d\mu_{\Sigma}$ the standard volume form corresponding to the induced metric on $\Sigma$. For more details regarding hypersurfaces and foliations, see Section 2.2 in \cite{paper2}.

\section{The TINP constant as a modified gradient flux}
\label{sec:AGeometricInterpretationOfTheTimeInvertedNewmanPenroseConstant}

Let $\Sigma$ be a Cauchy hypersurface which terminates at null infinity. Then the flux of the gradient vector field $\nabla \psi$ through $\Sigma$ is generically infinite:
First note that generically the following limits are infinite:
\[\int_{\Sigma}\nabla\psi\cdot n_{\Sigma}\,d\mu_{\Sigma}=\int_{\Sigma} n_{\Sigma}(\psi)\,d\mu_{\Sigma}=\infty,\]
where $n_{\Sigma}$ denotes the normal to $\Sigma$.\footnote{This follows from the fact that $\psi$ arising from smooth and compactly supported data will generically satisfy $\lim_{v\to \infty}v^3\partial_v(r\psi)|_{u=u'}\neq 0$ and $\lim_{v\to \infty}r\psi|_{u=u'} \neq 0$ for suitably large $u'$; see \cite{paper2}.} Furthermore, we clearly have that generically
\[\lim_{r_0\rightarrow \infty}\int_{\Sigma\cap\{r=r_0\}}\psi\,r^2d\omega=\lim_{r_0\rightarrow \infty}\int_{\Sigma\cap\{r=r_0\}}\psi r^2\,d\omega=\infty,\]
where we denote $d\omega=\sin\theta d\theta d\varphi$.

The following lemma shows that a combination of the above unbounded quantities is in fact bounded. 
\begin{lemma}
Consider a Cauchy hypersurface $\Sigma$ that crosses the future event horizon $\mathcal{H}^{+}$ and terminates at future null infinity $\mathcal{I}^{+}$. We denote by $n_{\Sigma}$ the normal vector field to $\Sigma$. Let $\psi$ be a solution to the wave equation \eqref{we} with vanishing Newman--Penrose constant $I_{0}[\psi]=0$ such that in fact
\[\lim_{r\rightarrow \infty} r^3 \partial_v(r\psi)|_{\Sigma}<\infty.  \]
Then the following limit 
\[\lim_{r_0\rightarrow \infty}\left(\int_{\Sigma\cap\{r\leq r_0\}} n_{\Sigma}(\psi)\,d\mu_{\Sigma}+\int_{\Sigma\cap\{r=r_0\}}\psi\,r^2d\omega\right)\]
exists and is finite. 
\label{remarkderivation}
\end{lemma}
\begin{proof}
If $\Sigma=\mathcal{N}$ (equipped with the induced coordinate system $(v,\omega)$) is outgoing null then we let $n_{\mathcal{N}}:=\partial_v$, $d\mu_{\mathcal{N}}:= r^2 dvd\omega$
and we obtain
\begin{equation*}
\begin{split}
\lim_{r_0\rightarrow \infty}\left(\int_{\mathcal{N}\cap\{r\leq r_0\}} n_{\mathcal{N}}(\psi)\,d\mu_{\mathcal{N}}+\int_{\mathcal{N}\cap\{r=r_0\}}\psi\,r^2d\omega\right)
=&\: \lim_{r_0\rightarrow \infty} \int_{\mathcal{N}\cap\{r\leq r_0\}} \partial_v\psi\cdot r^2 dv d\omega\\
&+\int_{\mathcal{N}\cap\{r=r_0\}}\psi r ^2 d\omega\\
=&\: \lim_{r_0\rightarrow \infty}\int_{\mathcal{N}\cap\{r\leq r_0\}}\Big( \partial_v\psi\cdot r^2+\partial_{v}(r^2\psi)\Big) dv d\omega\\
 =&\: \lim_{r_0\rightarrow \infty} \int_{\mathcal{N}\cap\{r\leq r_0\}} 2\partial_v(r\psi)r dv d\omega<\infty,
\end{split}
\end{equation*}
since $rv^2|\partial_v(r\psi)|= O(1)$ by assumption. 

For a general spherically symmetric hyperboloidal hypersurface $\Sigma$, equipped with the induced coordinate system $(v,\omega)$, we have 
\[\partial_{v}^{\Sigma}=\partial_v-f\partial_u,\]
\[n_{\Sigma}= \frac{1}{\sqrt{fD}}\cdot(\partial_v+f\partial_u),\]
\[d\mu_{\Sigma}=\sqrt{g_{\Sigma}}dvd\omega =\sqrt{fD}r^2 dv d\omega.\]
where $f:\Sigma\rightarrow \mathbb{R}$ is defined in Section \ref{geosetting} and satisfies $f(v)=O(v^{-1-\epsilon})$.
Then,
\begin{equation*}
\begin{split}
&\lim_{r_0\rightarrow \infty}\left(\int_{\Sigma\cap\{r\leq r_0\}} n_{\Sigma}(\psi)\,d\mu_{\Sigma}+\int_{\Sigma\cap\{r=r_0\}}\psi\,r^2d\omega\right)\\
&=\lim_{r_0\rightarrow \infty}\left(\int_{\Sigma\cap\{r\leq r_0\}}  \frac{1}{\sqrt{fD}}\cdot(\partial_v\psi+f\partial_u\psi)\sqrt{fD}r^2dvd\omega+\int_{\Sigma\cap\{r=r_0\}}\psi r ^2 d\omega\right)
\\&= \lim_{r_0\rightarrow \infty}\int_{\Sigma_0\cap\{r\leq r_0\}}
\Big((\partial_v\psi+f\partial_u\psi)r^2 +\partial_{v}^{\Sigma}(r^2\psi)\Big)dvd\omega
\\&= \lim_{r_0\rightarrow \infty}\int_{\Sigma_0\cap\{r\leq r_0\}}
\Big(r^2\partial_v\psi+fr^2\partial_u\psi +\partial_{v}(r^2\psi)-f\partial_u (r^2\psi)\Big)dvd\omega\\
& =\lim_{r_0\rightarrow \infty} \int_{\Sigma\cap\{r\leq r_0\}}\Big( 2\partial_v(r\psi)r +fD\cdot (r\psi)\Big)dv d\omega<\infty 
\end{split}
\end{equation*}
since $|v\psi|, v^{1+\epsilon}fD, v^3|\partial_v(r\psi)|= O(1)$ along $\Sigma$.
\end{proof}

\begin{proof}[Proof of Theorem \ref{mytheo1}]
We will show the theorem in the case where the hypersurface $\Sigma$ is of spacelike-null type (see \cite{paper2}). The computation is identical for general hypersurfaces crossing $\mathcal{H}^{+}$ and intersecting $\mathcal{I}^{+}$. For the spacelike-null case, according to the computation in \cite{paper2}, we have
\begin{equation}
\begin{split}
4\pi I_0[\psi^{(1)}]=& -\lim_{r_0\to \infty}\int_{\Sigma\cap \{r=r_0\}}r^3\partial_{r}\phi \, d\omega+M\int_{\Sigma\cap \{r=R\}} r(2-Dh_{\Sigma})\phi\, d\omega\\
&+2M\int_{\Sigma\cap \{r\geq R\}}r \partial_v\phi\,dv'd\omega\\
&-M\int_{\Sigma \cap \{r\leq R\}} \Big(2(1-h_{\Sigma}D)r\partial_{\rho}\phi-(2-Dh_{\Sigma})rh_{\Sigma} T\phi-(r\cdot (Dh_{\Sigma})')\cdot\phi\Big)\,d\rho'd\omega,
\end{split}
\label{nptimeintegralformula}
\end{equation}
where $\phi=r\psi$ and $\partial_r=\frac{2}{D}\partial_v$.

We will show  \eqref{geomiequa}.
Consider first the spacelike piece $\Sigma_{R}=\Sigma\cap \{r\leq R\}$. We equip $\Sigma_{R}$ with the coordinate system $(\rho,\omega)$, where $\rho=r|_{\Sigma_R}$. Recall that the radial tangential vector field $\partial_{\rho}$ is given by 
\[\partial_{\rho}=-2D^{-1}\partial_u+h_{\Sigma}T \]
and hence
\[g_{\Sigma_R}(\partial_{\rho},\partial_{\rho})=2h_{\Sigma}-h_{\Sigma}^2D>0.\]
We therefore have that
\[d\mu_{\Sigma_R}=\sqrt{2h_{\Sigma}-h_{\Sigma}^2D}\cdot r^2 drd\omega.\]
The unit future-directed normal vector field $n_{\Sigma_R}$ is given by
\[n_{\Sigma_R}= \frac{1}{\sqrt{2h_{\Sigma}-h_{\Sigma}^2D}}\cdot\Big((h_{\Sigma}D-1)\partial_{\rho}+(2-h_{\Sigma}D)h_{\Sigma}T\Big), \]
so we obtain
\begin{equation*}
\begin{split}
\int_{\Sigma_R}n_{\Sigma_R}\psi \,d\mu_{\Sigma}=\int_{\Sigma_R}\Big(r^2\cdot(h_{\Sigma}D-1)\cdot \partial_{\rho}\psi+r^2\cdot(2-h_{\Sigma}D)h_{\Sigma}\cdot T\psi\Big)\Big)d \rho d\omega.
\end{split}
\end{equation*}
Consider now $r_0>R$. Then
\begin{align*}
\int_{\Sigma\cap\{R\leq r\leq r_0\}}n_{\Sigma}\psi \,d\mu_{\Sigma}=&\:\int_{\Sigma\cap\{R\leq r\leq r_0\}}\partial_v\psi \cdot r^2 dv d\omega,\\
\int_{\Sigma\cap \{r=r_0\}}\psi r^2 d\omega =&\:\int_{\Sigma\cap \{r=R\} }\psi r^2 d\omega  +\int_{\Sigma\cap\{R\leq r\leq r_0\}} \partial_v(r^2\psi) \, dv d\omega, \\
\int_{\Sigma\cap\mathcal{H}^{+}}\psi r^2 d\omega =&\: \int_{\Sigma\cap \{r=R \}}\psi r^2 d\omega- \int_{\Sigma_R}\partial_{\rho}(r^2\psi) d \rho d\omega,\quad \textnormal{and}\\
\int_{\Sigma\cap\{r=r_0\}}-\frac{2}{M}r\partial_v(r\psi)r^2 d\omega=&\:\int_{\Sigma\cap\{r=r_0\}}-\frac{D}{M}r^3\partial_r(r\psi) d\omega.
\end{align*}
By adding them up we obtain
\begin{equation*}
\begin{split}
G(\Sigma^{\leq r_0})[\psi]=&\int_{\Sigma_R}\Big(-\partial_{\rho}(r^2\psi) +r^2\cdot(h_{\Sigma}D-1)\cdot \partial_{\rho}\psi+r^2\cdot(2-h_{\Sigma}D)h_{\Sigma}\cdot T\psi\Big)d \rho d\omega\\ &+ \int_{\Sigma\cap\{R\leq r\leq r_0\}}\Big(\partial_v\psi \cdot r^2+ \partial_v(r^2\psi)\Big) dv d\omega+\int_{\Sigma\cap\{r=r_0\}}-\frac{2}{M}r\partial_v(r\psi)r^2 d\omega\\
&+2\int_{\Sigma\cap \{r=R\} }\psi r^2 d\omega,
\end{split}
\end{equation*}
and so
\begin{equation*}
\begin{split}
G(\Sigma^{\leq r_0})[\psi]=&\int_{\Sigma_R}\Big( -2 r^2 \partial_{\rho}\psi -2r\psi
+r^2 h_{\Sigma}D\partial_{\rho}\psi
 +r^2\cdot(2-h_{\Sigma}D)h_{\Sigma}\cdot T\psi\Big)drd\omega\\ 
&+ \int_{\Sigma\cap\{R\leq r\leq r_0\}}2r\partial_v(r\psi) dv d\omega+\int_{\Sigma\cap\{r=r_0\}}-\frac{D}{M}r^3\partial_r(r\psi) d\omega\\
&+2\int_{\Sigma\cap \{r=R\} }\psi r^2 d\omega.
\end{split}
\end{equation*}
If we denote
\[K= \int_{\Sigma_{R}}\Big( -2r^2\partial_{\rho}\psi-2r\psi+r^2h_{\Sigma}D\partial_{\rho}\psi\Big)\, drd\omega,\]
then, since $D=0$ on the horizon, we obtain
\begin{equation*}
\begin{split}
K=&K-\int_{\Sigma\cap \{r=R\}}DR^2h_{\Sigma}\psi \,d\omega+ \int_{\Sigma\cap \{r=R\}}DR^2h_{\Sigma}\psi \,d\omega\\
=&K-\int_{\Sigma\cap \{r=R\}}DR^2h_{\Sigma}\psi \,d\omega +\int_{\Sigma_{R}} \partial_{\rho}(r^2Dh_{\Sigma}\psi)
\,drd\omega\\
=&K-\int_{\Sigma\cap \{r=R\}}DR^2h_{\Sigma}\psi \,d\omega +\int_{\Sigma_{R}} \Big(h_{\Sigma}Dr^2\partial_{\rho}\psi +2h_{\Sigma}Dr\psi+r^2(Dh_{\Sigma})'\psi\Big)
\,drd\omega
\\=&-\int_{\Sigma\cap \{r=R\}}DR^2h_{\Sigma}\psi \,d\omega +\int_{\Sigma_{R}}\Big( 2(h_{\Sigma}D-1) r^2 \partial_{\rho}\psi +2(h_{\Sigma}D-1)r\psi +r^2\cdot (Dh_{\Sigma})'\cdot \psi\Big)\,dr d\omega
\\=&-\int_{\Sigma\cap \{r=R\}}DR^2h_{\Sigma}\psi \,d\omega +\int_{\Sigma_{R}}\Big( 2(h_{\Sigma}D-1) r\partial_{\rho}(r\psi)+r^2\cdot (Dh_{\Sigma})'\cdot \psi \Big)\,drd\omega.
\end{split}
\end{equation*}
Hence,
\begin{equation*}
\begin{split}
G(\Sigma^{\leq r_0})[\psi]=&\int_{\Sigma_R}\Big(  2(h_{\Sigma}D-1) r\partial_{\rho}(r\psi)+r^2\cdot (Dh_{\Sigma})'\cdot \psi 
 +r^2\cdot(2-h_{\Sigma}D)h_{\Sigma}\cdot T\psi\Big)drd\omega\\ 
&+ \int_{\Sigma\cap\{R\leq r\leq r_0\}}2r\partial_v(r\psi) dv d\omega+\int_{\Sigma\cap\{r=r_0\}}-\frac{D}{M}r^3\partial_r(r\psi) d\omega\\
&+\int_{\Sigma\cap \{r=R\} }\Big(2\psi R^2-DR^2h_{\Sigma}\psi\Big) d\omega.
\end{split}
\end{equation*}
Taking the limit of the above equation as $r_0\rightarrow +\infty$ yields the desired result. 
\end{proof}

\section{Conservation law for the modified gradient fluxes}
\label{sec:ConservationLawForTheModifiedGradientFluxes}

We show next that the expressions given by the right hand side of \eqref{geomiequa} are indeed conserved without invoking the time-inversion construction. Hence, we obtain a purely geometric interpretation of the time-inverted constants and their conservation law.
\begin{proposition} \textbf{(Conservation law for the TINP constant)}

Consider two arbitrary hypersurfaces $\Sigma_i, i=1,2$ which cross the event horizon and terminate at future null infinity.

Let $\psi$ be a solution to the wave equation \eqref{we} with vanishing Newman--Penrose constant $I_{0}[\psi]=0$ such that in fact
\[ \lim_{r\rightarrow \infty} r^3 |\partial_{v}(r\psi)|_{\Sigma_1}<\infty,\ \ \ \ \lim_{r\rightarrow \infty} r^3 |\partial_{v}(r\psi)|_{\Sigma_2}<\infty.  \] 
 Then, 
\begin{equation}
\lim_{r_0\rightarrow \infty}G(\Sigma_1^{\leq r_0})[\psi]=\lim_{r_0\rightarrow \infty}G(\Sigma_2^{\leq r_0})[\psi],
\label{geomiequacon}
\end{equation}
where $G$ is given by \eqref{trunc}. 
In other words, the expression $\lim_{r_0\rightarrow \infty}G(\Sigma^{\leq r_0})[\psi]$ is \textbf{independent} of the hypersurface $\Sigma$.
\label{prop2deri}
\end{proposition}
\begin{proof}
 We assume that $\Sigma_2$ lies in the causal future of $\Sigma_1$.

Let $v_{\infty}>0$ and large and consider the region $\mathcal{R}_{v_{\infty}}$ bounded by the hypersurfaces $\mathcal{H}^{+},\Sigma_1,\Sigma_2, \{v=v_{\infty}\}$. 
	\begin{figure}[H]
\begin{center}
\includegraphics[width=6.5cm]{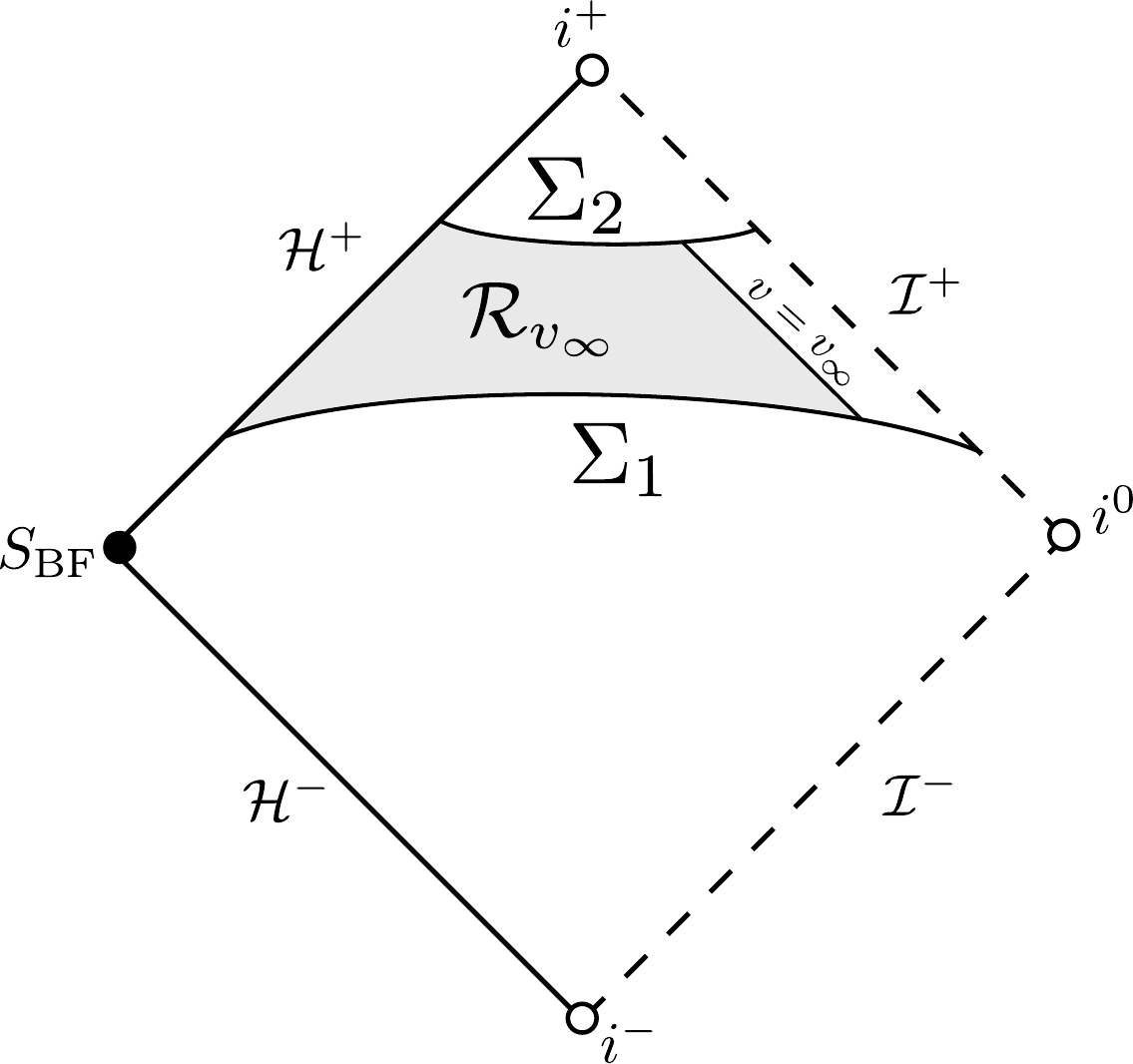}
\caption{\label{fig:14}The region $\mathcal{R}_{v_{\infty}}$.}
\end{center}
\end{figure}
\vspace{-0.6cm}
We apply Stokes' theorem for the gradient vector field 
\begin{equation}
X_{\mu}=\big(\nabla\psi)_{\mu}
\label{grad}
\end{equation}
in the region $\mathcal{R}_{v_{\infty}}$:
\begin{equation}
\begin{split}
\int_{\Sigma_1\cap \mathcal{R}_{v_{\infty}}} n_{\Sigma_1}\psi\, d\mu_{\Sigma_1}=\int_{\Sigma_2\cap \mathcal{R}_{v_{\infty}}} n_{\Sigma_2}\psi\, d\mu_{\Sigma_1}+\int_{\mathcal{H}^{+}\cap \mathcal{R}_{v_{\infty}}}T\psi\,r^2d\omega dv+\int_{\{v=v_{\infty}\}\cap \mathcal{R}_{v_{\infty}}}\partial_u\psi\,r^2d\omega du. 
\end{split}
\label{stokes1}
\end{equation}
Note that since $\mathcal{H}^{+}$ is Killing horizon with normal $T$:
\begin{equation}
\begin{split}
\int_{\mathcal{H}^{+}\cap \mathcal{R}_{v_{\infty}}}T\psi\,r^2d\omega=&\int_{\mathcal{H}^{+}\cap \Sigma_2}\psi\,r^2d\omega-\int_{\mathcal{H}^{+}\cap \Sigma_1}\psi\,r^2d\omega.
\end{split}
\label{horizonkillingintegration}
\end{equation}
Furthermore, we obtain
\begin{equation}
\begin{split}
\int_{\{v=v_{\infty}\}\cap \mathcal{R}_{v_{\infty}}}\partial_u\psi\,r^2d\omega du =&\int_{\{v=v_{\infty} \}\cap \Sigma_2}\psi\,r^2d\omega -\int_{\{v=v_{\infty} \}\cap \Sigma_1}\psi\,r^2d\omega -\int_{\{v=v_{\infty}\}\cap \mathcal{R}_{v_{\infty}}}\psi \cdot  2r \cdot \partial_u r \,du d\omega
\\=&
\int_{\{v=v_{\infty} \}\cap \Sigma_2}\psi\,r^2d\omega -\int_{\{v=v_{\infty} \}\cap \Sigma_1}\psi\,r^2d\omega +\int_{\{v=v_{\infty}\}\cap \mathcal{R}_{v_{\infty}}}\phi \cdot D \,du d\omega
\end{split}
\label{integrationbyparts}
\end{equation}
where $\phi=r\psi$. The wave equation for the spherically symmetric $\psi$ takes the form 
\[\partial_u \partial_v \phi= A\cdot \phi \]
where 
\begin{equation}
A(r)= -\frac{1}{4r}D\cdot D'= -\frac{M}{2}\frac{1}{r^3}+ O\left( r^{-4}\right),
\label{eqa}
\end{equation}
since 
\begin{equation}
D(r)= 1-\frac{2M}{r}+O\left(r^{-2}\right).
\label{eqd}
\end{equation}
Hence, 
\begin{equation*}
\begin{split} \partial_u (r^3 \partial_v \phi)=& 3r^2 \partial_u r \partial_v\phi +r^3\partial_u\partial_v \phi\\
=& -\frac{3D}{2}r^2  \partial_v\phi +r^3\cdot A\cdot  \phi,\\
\end{split}
\end{equation*}
which, in view of \eqref{eqa},\eqref{eqd}, yields
\begin{equation}
D\cdot \phi=\partial_u\left(-\frac{2}{M}r^3 \partial_v \phi\right)+\left(-\frac{3}{M} + O\left(r^{-1}\right)\right)\cdot r^2\partial_v\phi + \phi \cdot  O\left(r^{-1} \right).
\label{eq:phi}
\end{equation}
Therefore, 
\begin{equation}
\begin{split}
\int_{\{v=v_{\infty}\}\cap \mathcal{R}_{v_{\infty}}}\phi \cdot D \,du d\omega
=& \int_{\{v=v_{\infty} \}\cap \Sigma_2}-\frac{2}{M}r^3\partial_v\phi\,d\omega -\int_{\{v=v_{\infty} \}\cap \Sigma_1}-\frac{2}{M}r^3\partial_v\phi\,d\omega \\&+ \int_{\{v=v_{\infty}\}\cap \mathcal{R}_{v_{\infty}}} O\left(r^{-1}\right)\cdot\left( r^2\partial_v\phi + \phi\right) \, du d\omega
\\& +  \int_{\{v=v_{\infty}\}\cap \mathcal{R}_{v_{\infty}}}-\frac{3}{M} \cdot r^2\partial_v\phi \, du d\omega.
\end{split}
\label{integralradiation}
\end{equation}
Therefore, in view of \eqref{stokes1}, \eqref{horizonkillingintegration}, \eqref{integrationbyparts}, \eqref{integralradiation}, we obtain
\begin{equation*}
\begin{split}
&\int_{\Sigma_2\cap \mathcal{R}_{v_{\infty}}}n_{\Sigma_2}\psi\,d\mu_{\Sigma_2} +\int_{\mathcal{H}^{+}\cap\Sigma_2}\psi\,r^2d\omega+\int_{\{v=v_{\infty}\}\cap \Sigma_2 }\psi\,r^2d\omega- \int_{\{v=v_{\infty} \}\cap \Sigma_2}\frac{2}{M}r\partial_v\phi\,r^2d\omega\\=& 
\int_{\Sigma_1\cap \mathcal{R}_{v_{\infty}}}n_{\Sigma_1}\psi\,d\mu_{\Sigma_1} +\int_{\mathcal{H}^{+}\cap\Sigma_1}\psi\,r^2d\omega +\int_{\{v=v_{\infty}\}\cap \Sigma_1 }\psi\,r^2d\omega- \int_{\{v=v_{\infty} \}\cap \Sigma_1}\frac{2}{M}r\partial_v\phi\,r^2d\omega\\&+\mathcal{E}^{\Sigma_1,\Sigma_2}_{{v_{\infty}}}[\psi],
\end{split}
\end{equation*}
where
\begin{equation}
\begin{split}
\mathcal{E}^{\Sigma_1,\Sigma_2}_{{v_{\infty}}}[\psi]=\int_{\{v=v_{\infty}\}\cap \mathcal{R}_{v_{\infty}}} O\left(r^{-1}\right)\cdot\left( r^2\partial_v\phi + \phi\right) \, du d\omega
 +  \int_{\{v=v_{\infty}\}\cap \mathcal{R}_{v_{\infty}}}-\frac{3}{M} \cdot r^2\partial_v\phi \, du d\omega.
\end{split}
\label{preliconse}
\end{equation}
Then, recalling \eqref{trunc}, \eqref{preliconse} yields
\begin{equation}
G(\Sigma_{2}^{\leq r_{\infty}^2})[\psi]=G(\Sigma_{1}^{\leq r_{\infty}^{1}})[\psi]+\mathcal{E}^{\Sigma_1,\Sigma_2}_{{v_{\infty}}}[\psi]
\label{conse1}
\end{equation}
where $r_{\infty}^i:=r\big(\Sigma_{i}\cap \{v=v_{\infty}\}\big)$. We have already established in Lemma \ref{remarkderivation} that the limit 
\[\lim_{r_0\rightarrow \infty}G(\Sigma_{i}^{\leq r_0})[\psi]<\infty.\]
Furthermore, since we have $r^3|\partial_v(r\psi)|\leq C$ we have 
\[\lim_{v_{\infty}\rightarrow \infty}\mathcal{E}^{\Sigma_1,\Sigma_2}_{{v_{\infty}}}[\psi]=0. \]
Hence, we can take the limit of \eqref{conse1} as $v_{\infty}\rightarrow \infty$ (and hence as $r_{\infty}^1, r_{\infty}^2\rightarrow \infty$) to obtain
\[\lim_{r_0\rightarrow \infty}G(\Sigma_{2}^{\leq r_0})[\psi]=\lim_{r_0\rightarrow \infty}G(\Sigma_{2}^{\leq r_0})[\psi]\]
which is the desired result. 
\end{proof}
In \cite{paper2} we defined the TINP constant for all (compactly supported) smooth initial data on a hypersurface crossing the event horizon to the future of the bifurcate sphere; see also Section \ref{sec:ReviewOfTheTimeIntegralConstruction}. As illustrated in Section \ref{sec:ReviewOfTheTimeIntegralConstruction}, this definition breaks down for Cauchy hypersurfaces emanating from the bifurcation sphere since the time-integral construction is singular for smooth initial data with non-trivial support on the bifurcate sphere. 

The following proposition, an immediate corollary of the divergence identity, establishes a generalized conservation law and hence \emph{allows us to extend the definition of the TINP constant with respect to initial data on hypersurfaces which pass through the bifurcate sphere}.\footnote{Clearly, the gradient flux is always defined for such hypersurfaces through the bifurcation sphere.}
\begin{proposition}
Let $\Sigma_{\text{BF}}$ be a Cauchy hypersuface which emanates from the bifurcate sphere $S_{\text{BF}}$ and terminates at null infinity. Let $\Sigma_{0}$ be a a Cauchy hypersuface which crosses the event horizon and terminates at null infinity. We assume 
\[ \Sigma_{\text{BF}}\cap \{r\geq R \}= \Sigma_{0}\cap \{r\geq R \}\] 
for some large $R>0$. Then, 
\begin{equation}
\int_{\Sigma_{0}\cap \{r\leq R\}}n_{\Sigma_{0}}\psi\,d\mu_{\Sigma_0} +\int_{\Sigma_{0}\cap \mathcal{H}^{+}}\psi\,r^2d\omega=\int_{\Sigma_{\text{BF}}\cap \{r\leq R\}}n_{\Sigma_{\text{BF}}}\psi\,d\mu_{\Sigma_{\text{BF}}}  +\int_{\Sigma_{\text{BF}}\cap S_{\text{BF}}}\psi\,r^2d\omega.
\label{birequaprop}
\end{equation}
\label{bifuprop}
\end{proposition}
\begin{proof}
Let $\mathcal{B}$ denote the region bounded by $\Sigma_{\text{BF}}\cap \{r\leq R\}, \Sigma_{0}\cap \{r\leq R\}$ and $\mathcal{H}^{+}$. 
	\begin{figure}[H]
\begin{center}
\includegraphics[width=6.5cm]{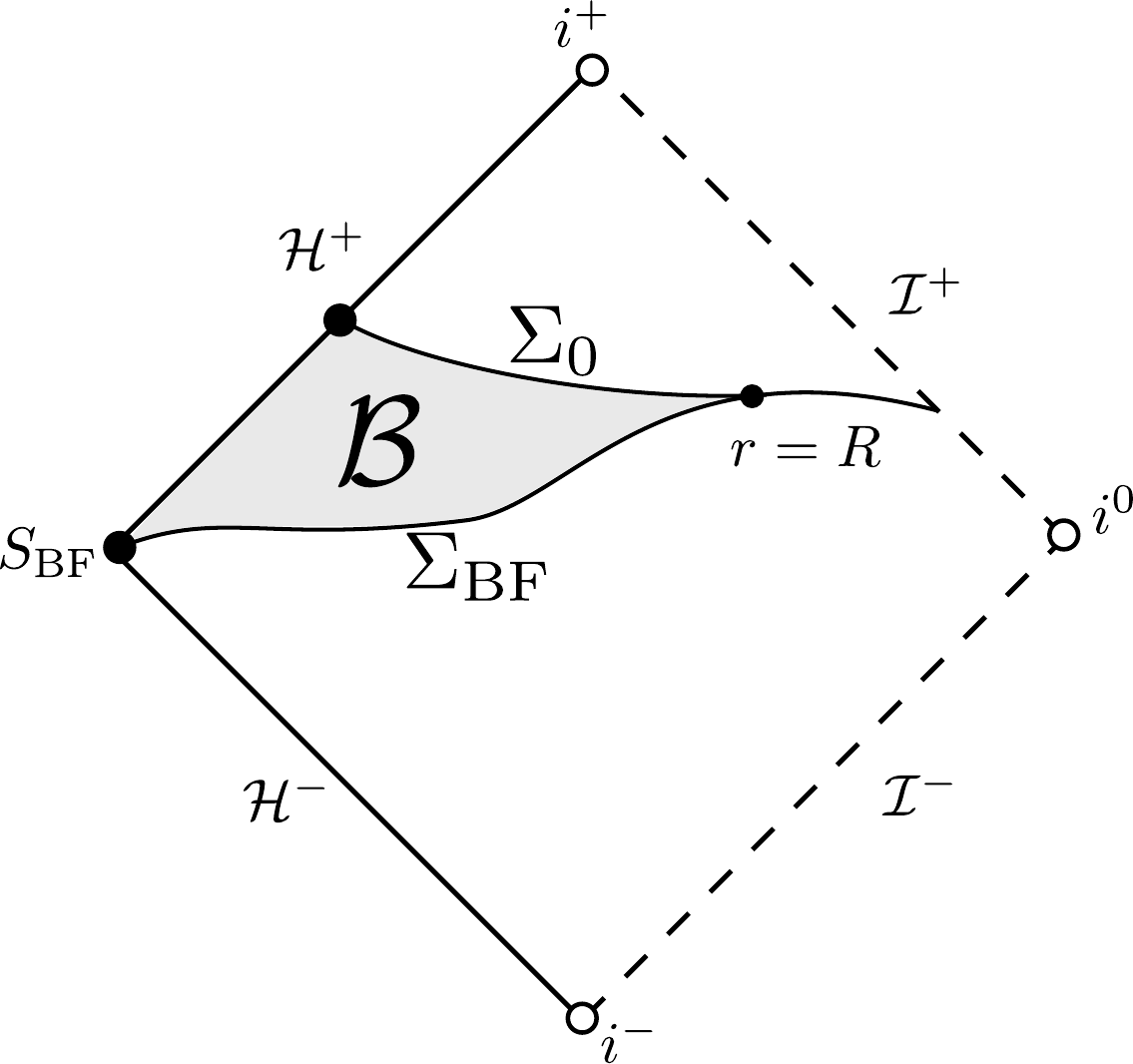}
\caption{\label{fig:14b}The region $\mathcal{B}$.}
\end{center}
\end{figure}
\vspace{-0.6cm}
Let $V$ be a regular (spherically symmetric) null  vector field normal to the future event horizon such that 
\[ V\Big|_{\text{BF}}\neq 0.\]
Let also $x$ denote the unique smooth function on $\mathcal{H}^{+}$ such that 
\[V(x)=1, \ \ \ x\big|_{\text{BF}}=0.\]
The volume form $d\mu_{\mathcal{H}^{+}}$ on $\mathcal{H}^{+}$ expressed in the coordinate system $(x,\omega)$ is given by
\[ d\mu_{\mathcal{H}^{+}}=(r_{\mathcal{H}^{+}})^2\, dx d\omega, \]
where $r_{\mathcal{H}^{+}}$ is the area-radius of the event horizon. Note that since $\mathcal{H}^{+}$ is a Killing horizon, the area-radius $r_{\mathcal{H}^{+}}$ is constant. 
 
We apply the divergence identity for the gradient vector field \eqref{grad} in region $\mathcal{B}$ to obtain:
\begin{equation}
\begin{split}
\label{eq:stokeswavebf}
\int_{\mathcal{B}} \square_g\psi \, d\mu =&\:
\int_{\Sigma_{\text{BF}}\cap \{r\leq R\}}n_{\Sigma_{\text{BF}}}(\psi)\, d\mu_{\Sigma_{\text{BF}}}
\\&-\int_{\Sigma_{0}\cap \{r\leq R\}}n_{\Sigma_{0}}(\psi)\, d\mu_{\Sigma_{0}}-\int_{\mathcal{H}^{+}\cap \mathcal{B}}V\psi \, d\mu_{\mathcal{H}^{+}},
\end{split}
\end{equation}
where $d\mu$ is the standard volume form corresponding to the metric $g$.

Note that 
\begin{equation}
\begin{split}
\int_{\mathcal{H}^{+}\cap \mathcal{B}}V\psi \, d\mu_{\mathcal{H}^{+}}=&
\int_{\s^2}\left(\int_{x}V\psi \, dx\right) r_{\mathcal{H}^{+}}^2 \,d\omega\\
=& \int_{\s^2}\left( \psi\Big|_{\mathcal{H}^+\cap \Sigma_{0}}-\psi\Big|_{\mathcal{H}^+\cap \Sigma_{ \text{BF}}}\right) r_{\mathcal{H}^{+}}^2 \,d\omega\\
=& \int_{\mathcal{H}^+\cap \Sigma_{0}}\psi\, r^2d\omega-\int_{\mathcal{H}^+\cap \Sigma_{\text{BF}}}\psi\, r^2d\omega.
\end{split}
\label{bifuconclu}
\end{equation}
The desired result follows from \eqref{we}, \eqref{eq:stokeswavebf} and \eqref{bifuconclu}. 
\end{proof}

\section{The TINP constant on $\{t=0\}$}
\label{sec:TheTINPConstantOnT0}

We now have all the tools to prove Theorem \ref{mytheo3}.
\begin{proof}[Proof of Theorem \ref{mytheo3}]
We use Theorem \ref{mytheo1} to express the TINP constant $I_0^{(1)}[\psi]$ constant in terms of the modified gradient flux, given by \eqref{trunc} and \eqref{geomiequa}. We next use the conservation law derived in Theorem \ref{mytheo2} to obtain the value of $I_0^{(1)}[\psi]$ in terms of initial data on the hypersurface 
\[\Sigma=\Big\{\{t=0\}\cap \{r\leq R \}\Big\}\cup \Big\{ \{u=u_0\}\cap \{r\geq R\}\Big\}\] 
as follows:
\[\frac{4\pi}{M} I_0^{(1)}[\psi]=\int_{S_{\text{BF}}}\!\!\psi\,r^2d\omega+\int_{\Sigma} n_{\Sigma}(\psi)\,d\mu_{\Sigma}+\lim_{r_0\rightarrow \infty}\int_{\Sigma\cap\{r=r_0\}}\Big(\psi-\frac{2}{M}r\partial_v(r\psi)\Big)\,r^2d\omega.\]
Hence, for initial data on $\{t=0\}$ supported in $\{r\leq R\}$ the third term on the right vanishes. The proof of Theorem \ref{mytheo3} follows from the identity \eqref{nt0}.

\end{proof}


\small

Department of Mathematics, University of California, Los Angeles, CA 90095, United States, yannis@math.ucla.edu

\bigskip

Princeton University, Department of Mathematics, Fine Hall, Washington Road, Princeton, NJ 08544, United States, aretakis@math.princeton.edu

\bigskip

Department of Mathematics, University of Toronto Scarborough 1265 Military Trail, Toronto, ON, M1C 1A4, Canada, aretakis@math.toronto.edu

\bigskip

Department of Mathematics, University of Toronto, 40 St George Street, Toronto, ON, Canada, aretakis@math.toronto.edu

\bigskip

Department of Pure Mathematics and Mathematical Sciences, University of Cambridge, Wilberforce Road, Cambridge CB3 0WB, United Kingdom, dg405@cam.ac.uk

\end{document}